\documentclass[sigconf, nonacm]{acmart}

\newcommand\vldbdoi{XX.XX/XXX.XX}
\newcommand\vldbpages{XXX-XXX}
\newcommand\vldbvolume{14}
\newcommand\vldbissue{1}
\newcommand\vldbyear{2020}
\newcommand\vldbauthors{\authors}
\newcommand\vldbtitle{\shorttitle} 
\newcommand\vldbavailabilityurl{URL_TO_YOUR_ARTIFACTS}
\newcommand\vldbpagestyle{plain} 

\usepackage{balance}

\usepackage[utf8]{inputenc}

\usepackage{caption}
\usepackage{subcaption}

\usepackage{natbib}
\usepackage{graphicx}
\usepackage{enumitem}
\usepackage{todonotes}
\usepackage{xcolor}

\usepackage[ruled,vlined,linesnumbered]{algorithm2e} 
\usepackage{tikz}
\usepackage{tikzscale}
\definecolor{darkgreen}{RGB}{0,104,0}

\usepackage{amsthm}
\usepackage{amsmath}
\usepackage{breqn}

\usepackage{multirow}

\newtheorem{theorem}{Theorem}[section]

\newtheorem{lemma}[theorem]{Lemma}
\theoremstyle{definition}
\newtheorem{definition}{Definition}[section]
\theoremstyle{definition}
\newtheorem{example}{Example}[section]

\newcommand{\degout} {{\it{d^+_{out}} }}

\newcommand{\degin} {{\it{d^*_{in}}}}
\newcommand{\od}{\mathbb{O}}  
\newcommand{\vgr}{{\it{V^+ \setminus V^*}}}
\newcommand{\vgrratio}{{\it{|V^+| / |V^*|}}}
\newcommand{\fs}{\small} 
\newcommand{\adj}{{\it adj}}
\newcommand{\Deg}{{\it deg}} 
\newcommand{\dequeue}{{\it dequeue}}
\newcommand{\enqueue}{{\it enqueue}}
\newcommand{\core}{{\it{core}}}

\newcommand{\post}{{\it{post}}}
\newcommand{\pre}{{\it{pre}}}
\newcommand{\mcd}{{\it{mcd}}}
\newcommand{\DEG}{{\it{Deg}}} 
\newcommand{\nlb}{{\it{\#lb}}}
\newcommand{\nrp}{{\it{\#rp}}}

\usepackage{graphics}
\graphicspath{{./fig/}}

\begin{document}
\pagestyle{plain} 

\title{Simplified Algorithms for Order-Based Core Maintenance}
\date{\today}

\iffalse
\author{Ben Trovato}
\affiliation{%
  \institution{Institute for Clarity in Documentation}
  \streetaddress{P.O. Box 1212}
  \city{Dublin}
  \state{Ireland}
  \postcode{43017-6221}
  \country{France}
}
\email{trovato@corporation.com}

\author{Lars Th{\o}rv{\"a}ld}
\orcid{0000-0002-1825-0097}
\affiliation{%
  \institution{The Th{\o}rv{\"a}ld Group}
  \streetaddress{1 Th{\o}rv{\"a}ld Circle}
  \city{Hekla}
  \country{Iceland}
}
\email{larst@affiliation.org}

\else

\author{Bin Guo}
\affiliation{%
  \institution{McMaster University}
  \city{Hamilton}
  \state{Ontario}
  \country{Canada}
}
\email{guob15@mcmaster.ca}

\author{Emil Sekerinski}
\affiliation{%
  \institution{McMaster University}
  \city{Hamilton}
  \state{Ontario}
  \country{Canada}
}
\email{emil@mcmaster.ca}

\fi

\begin{abstract}
Graph analytics attract much attention from both research and industry communities. Due to the linear time complexity, the $k$-core decomposition is widely used in many real-world applications such as biology, social networks, community detection, ecology, and information spreading. 
In many such applications, the data graphs continuously change over time. The changes correspond to edge insertion and removal. Instead of recomputing the $k$-core, which is time-consuming, we study how to maintain the $k$-core efficiently. That is, when inserting or deleting an edge, we need to identify the affected vertices by searching for more vertices. 
The state-of-the-art order-based method maintains an order, the so-called $k$-order, among all vertices, which can significantly reduce the searching space. 
However, this order-based method is complicated for understanding and implementation, and its correctness is not formally discussed.
In this work, we propose a simplified order-based approach by introducing the classical Order Data Structure to maintain the $k$-order, which significantly improves the worst-case time complexity for both edge insertion and removal algorithms. 
Also, our simplified method is intuitive to understand and implement; it is easy to argue the correctness formally. 
Additionally, we discuss a simplified batch insertion approach.
The experiments evaluate our simplified method over 12 real and synthetic graphs with billions of vertices. 
Compared with the existing method, our simplified approach achieves high speedups up to 7.7x and 9.7x for edge insertion and removal, respectively. 
\end{abstract}

\maketitle

\pagestyle{\vldbpagestyle}
\begingroup\small\noindent\raggedright\textbf{PVLDB Reference Format:}\\
\vldbauthors. \vldbtitle. PVLDB, \vldbvolume(\vldbissue): \vldbpages, \vldbyear.\\
\href{https://doi.org/\vldbdoi}{doi:\vldbdoi}
\endgroup
\begingroup
\renewcommand\thefootnote{}\footnote{\noindent
This work is licensed under the Creative Commons BY-NC-ND 4.0 International License. Visit \url{https://creativecommons.org/licenses/by-nc-nd/4.0/} to view a copy of this license. For any use beyond those covered by this license, obtain permission by emailing \href{mailto:info@vldb.org}{info@vldb.org}. Copyright is held by the owner/author(s). Publication rights licensed to the VLDB Endowment. \\
\raggedright Proceedings of the VLDB Endowment, Vol. \vldbvolume, No. \vldbissue\ %
ISSN 2150-8097. \\
\href{https://doi.org/\vldbdoi}{doi:\vldbdoi} \\
}\addtocounter{footnote}{-1}\endgroup

\ifdefempty{\vldbavailabilityurl}{}{
\vspace{.3cm}
\begingroup\small\noindent\raggedright\textbf{PVLDB Artifact Availability:}\\
The source code, data, and/or other artifacts have been made available at \url{https://github.com/Itisben/SimplifiedCoreMaint}.
\endgroup
}

\section{Introduction} 


Given an undirected graph $G=(V,E)$, the $k$-core decomposition is to identify the maximal subgraph $G'$ in which each vertex has a degree of at least $k$; the core number of each vertex $u$ is defined as the maximum value of $k$ such that $u$ is contained in the $k$-core of $G$~\cite{bz2003,Kong2019}. 
It is well-known that the core numbers can be computed with linear running time $O(|V|+|E|)$ \cite{bz2003}.
Due to the linear time complexity, the $k$-core decomposition is easily and widely used in many real-world applications. 
In~\cite{Kong2019}, Kong et al.~summarize a large number of applications in biology, social networks, community detection, ecology, information spreading, etc. Especially in~\cite{burleson2020k}, Lesser et al.~investigate the $k$-core robustness in ecological and financial networks.

In a survey~\cite{Malliaros2020}, Malliaros et al.~summarize the main research work related to $k$-core decomposition from 1968 to 2019.
In static graphs, the computation of the core numbers has been extensively studied \cite{bz2003,cheng2011efficient,khaouid2015k,montresor2012distributed,wen2016efficient}. 
However, in many real-world applications, such as determining the influence of individuals in spreading epidemics in dynamic complex networks~\cite{miorandi2010k} and tracking the actual spreading dynamics in dynamic social media networks~\cite{pei2014searching}, the data graphs continuously change over time. The changes correspond to the insertion and deletion of edges, which may have an impact on the core numbers of some vertices in the graph. 
Graphs of this kind are called \emph{dynamic graphs}. 
When inserting or removing an edge, it is time-consuming to recalculate the core numbers of all vertices; a better approach is first to find the affected vertices and then to update their corresponding core numbers. 
The problem of maintaining the core numbers for dynamic graphs is called \emph{core maintenance}. 
To the best of our knowledge, little work is done on the $k$-core maintenance~ \cite{Saryuce2016,Zhang2017,wu2015core,sariyuce2013streaming}. 

In this work, we focus on core maintenance. More formally, given an undirected dynamic graph $G=(V,E)$, after inserting an edge into or removing an edge from $G$, the problem is how to efficiently update the core number for the affected vertices. To do this, we first need to identify a set of vertices whose core numbers need to be updated (denoted as $V^*$) by traversing a possibly larger set of vertices (denoted as $V^+$). Then it is easy to re-compute the new core numbers of vertices in $V^*$.
In~\cite{Zhang2019}, Zhang et al.~prove that the core maintenance is asymmetric: the edge removal is bounded for $V^*=V^+$, but the edge insertion is unbounded for $V^*\subseteq V^+$.
In other words, to identify $V^*$, the edge removal only needs to traverse $V^*$; however, the edge insertion may traverse a much larger set of vertices than $V^*$. 
In practice, an edge removal algorithm for core maintenance~\cite{Saryuce2016,Zhang2017} is easy to devise; but for edge insertion, it is challenging. 
Clearly, an efficient edge insertion algorithm should have a small cost for identifying $V^*$, which means a small ratio $|V^+|/|V^*|$. In this work, we mainly discuss the edge insertion algorithms for core maintenance.

In \cite{Saryuce2016}, Sariy{\"u}ce et al.~propose a traversal algorithm. This insertion algorithm searches $V^*$ only in a local region near the edge that is inserted, which can be much faster than recomputing the core numbers for the whole graph. However, this insertion algorithm has a high variation in terms of performance due to the high variation of the ratio $|V^+|/|V^*|$. 
In \cite{Zhang2017}, Zhang et al. propose an \emph{order-based approach}, which is the state-of-the-art method for core maintenance. 
The main idea is that a $k$-order is explicitly maintained among vertices such that $u \preceq v$ for every two vertices in a graph $G$. Here, a $k$-order, $(v_1 \preceq v_2 \preceq .... \preceq v_n)$, for each vertex~$v_i$ in a graph $G$, is an order that the core number determined by a core decomposition algorithm, the BZ algorithm \cite{bz2003}. 
When a new edge $(u, v)$ is inserted, the potentially affected vertices are checked with such $k$-order, by which numerous vertices are avoided to be checked.
In this case, the size of $V^+$ is greatly reduced and so that the ratio $|V^+|/|V^*|$ is typically much smaller and has less variation compared with the traversal algorithm. 
Thus the computation time is significantly improved. 

However, this order-based approach has two drawbacks.
First, the order-based edge insertion algorithm is so complicated that it is not intuitive for easy understanding. This complexity further brings difficulties to the correctness and implementation; actually, the proof of correctness for the edge-insert algorithm is not formally discussed in \cite{Zhang2017}.
Second, the $k$-order of the vertices in a graph is maintained by two specific data structures: 1) $\mathcal A $ (double linked lists combined with balanced binary search trees) for operations like inserting, deleting, comparing the order of two vertices, all of which requires worst-case $O(\log |V|)$ time; and 2) $\mathcal B$ (double linked lists combined with heaps) for searching the ordered vertices by jumping unnecessary ones, which requires worst-case $O(\log |V|)$ time; both data structures are complicated to implement.

In this work, we try to overcome the above drawbacks in \cite{Zhang2017} by proposing our simplified order-based approach. 
The idea behind our new approach is that we introduce a well-known \emph{Order Data Structure} \cite{dietz1987two,bender2002two} to maintain the $k$-order of vertices in a graph $G$. 
By doing this, there are several benefits. 
First, this classical Order Data Structure only requires amortized $O(1)$ time for order operations, including inserting, deleting, and comparing the order of two vertices; this is faster than the $\mathcal A$ data structure in \cite{Zhang2017} especially when $|V|$ is large.  
Also, the original order-based insertion algorithm can be introduced to maintain each affected vertex in $k$-order in worst-case $O(log|E^+|)$ time ($|E^+|$ is the number of edges adjacent to vertices in $V^+$); this is also faster than the $\mathcal B$ data structure in~\cite{Zhang2017} since normally we have $|E^+| \ll |V|$.
Second, compared with the method in \cite{Zhang2017}, 
when introducing the Order Data Structures and priority queues, the $\mathcal A$ and $\mathcal B$ data structures can be abandoned and so that the order-based approach can be significantly simplified; also, our new approach simplifies the proof of correctness.
Finally, our simplified order-based insertion algorithm can be easily extended to handle a batch of insertion edges without difficulties since it is common that a great number of edges are inserted or removed simultaneously; by doing this, the vertices in $V^+\setminus V^*$ are possibly avoided to be repeatedly traversed so that the total size of $V^+$ is smaller compared to unit insertion.  
The main contributions are summarized as below:
\begin{itemize} [noitemsep,topsep=0pt] 
    \item We investigate the drawbacks of the state-of-the-art order-based core maintenance algorithms in \cite{Zhang2017}. 
    
    \item Based on \cite{Zhang2017}, by introducing the Order Data Structure \cite{dietz1987two,bender2002two}, we propose a simplified order-based insertion algorithm. 
    Not only can the worst-case time complexity be improved, but also the proof of correctness is simplified.

    \item We extend our simplified core insertion algorithm to handle a batch of edges, with smaller size of $V^*$ compared to unit insertion.
    
    \item Finally, we conduct extensive experiments with different kinds of real data graphs to evaluate different algorithms.
\end{itemize}

The rest of this paper is organized as follows. The preliminaries are given in Section 2. 
The original order-based algorithm is reviewed in Section 3. Our simplified order-based insertion and removal algorithms are proposed in Section 4. Our simplified order-based batch insertion is proposed in Section 5. 
Related work is discussed in Section 6. We report on extensive performance studies in Section 7 and conclude in Section 8.

\section{Preliminaries} \label{Preliminaries}

Let $G = (V, E)$ be an undirected unweighted graph, where $V(G)$ denotes the set of vertices and $E(G)$ represents the set of edges in $G$. When the context is clear, we will use $V$ and $E$ instead of $V(G)$ and $E(G)$ for simplicity, respectively.
Note that, as $G$ is an undirected graph, an edge $(u, v)\in E(G)$ is equivalent to $(v, u)\in E(G)$. 
We denote the number of vertices and edges of $G$ by $n$ and $m$, respectively. We define the set of neighbors of a vertex $u \in V$ as $u.\adj$, formally $u.\adj = \{v \in V: (u, v) \in E\}$.
We denote the degree of $u$ in $G$ as $u.\Deg = |u.\adj|$.
Also, to analyze the time complexity, we denote the maximal degree among all vertices in $G$ as $\DEG(G)=\max\{v\in V(G): v.\Deg\}$.
We say a graph $G'$ is a subgraph of $G$, denoted as $G'\subseteq G$, if $V(G') \subseteq V(G)$ and $E(G') \subseteq E(G)$. 
Given a subset $V' \subseteq V$, the subgraph induced by $V'$, denoted as $G(V')$, is defined as $G(V') = (V', E')$ where $E' = \{(u, v) \in E : u, v \in V'\}$.

\begin{definition} [$k$-Core]
Given an undirected graph $G=(V, E)$ and an integer $k$, a subgraph $G_k$ of $G$ is called a $k$-core if it satisfies the following conditions: (1) for $\forall u \in V(G_k)$, $u.\Deg \geq k$; (2) $G_k$ is maximal. Moreover, $G_{k+1} \subseteq G_k$, for all $k \geq 0$, and $G_0$ is just $G$. 
\end{definition}

\begin{definition}[Core Number]
\label{def:corenumber}
Given an undirected graph $G=(V,E)$, the core number of a vertex $u\in G(V)$, denoted as $u.\core$, is defined as $u.\core = max\{k: u \in V(G_k)\}$. That means $u.core$ is the largest $k$ such that there exists a $k$-core containing $u$.
\end{definition}

\begin{definition} [Subcore]
\label{def:subcore}
Given an undirected graph $G=(V,E)$, a maximal set of vertices $S\subseteq V$ is called a $k$-subcore if (1) $\forall u \in S, u.\core = k$; (2) the induced subgraph $G(S)$ is connected. The subcore that contains vertex $u$ is denoted as \texttt{sc}$(u)$.  
\end{definition}

\paragraph{Core Decomposition} 
Given a graph $G=(V,E)$, the problem of computing the core number for each $u \in V(G)$ is called core decomposition. In \cite{bz2003}, Batagelj and Zaversnik propose an algorithm with a liner running time of $O(m+n)$, the so-called BZ algorithm. 
The general idea is the \emph{peeling process}. That is, to compute the $k$-core $G_k$ of $G$, the vertices (and their adjacent edges) whose degrees are less than $k$ are repeatedly removed. When there are no more vertices to remove, the resulting graph is the $k$-core of $G$.

\begin{algorithm}[!htb]
\caption{BZ algorithm for core decomposition}
\label{alg:bz}
\small
\SetAlgoNoEnd
\DontPrintSemicolon
\SetKwInOut{Input}{input}\SetKwInOut{Output}{output}
\Input{an undirected graph $G=(V,E)$}
\Output{the core number $u.\core$ for each $u\in V$}
    
    \lFor{$u \in V$}{ 
        $u.d \gets |u.\adj|$; $u.core = \varnothing$}
    $Q\gets$ a min-priority queue by $u.d$ for all $u\in V$\;
    
    \While{$Q \neq \emptyset$}{
        $u \gets Q.\dequeue ()$\;
       $u.\core \gets u.d$; remove $u$ from $G$\; \label{alg:bz-core}
       \For{$v \in u.\adj$}{
            \lIf{$u.d < v.d$}{$v.d \gets v.d - 1$}
        }
        
        update $Q$\;
    }
\end{algorithm}

Algorithm \ref{alg:bz} shows the steps of the BZ algorithm. In initialization, for each vertex $u\in V$, the auxiliary degree $u.d$ is set to $|u.\adj|$ and the core number $u.core$ is not identified (line 1). The postcondition is that for each vertex $u\in V$, the $u.d$ equals to the core number, formally $u.d = u.\core$. We state informally lines 3 - 8 as a loop invariant: (1) the vertex $u$ always has the minimum degree $u.d$ since $u$ is removed from the min-priority queue $Q$ (line 4); and (2) if $u$ obtains its core number, $u.core$ equals to $u.d$ (line 5). 
The key step is updating $v.d$ for all $v\in u.adj$. That is, $v.d$ are decremented by $1$ if $u.d$ is smaller than $v.d$ (lines 6 and 7). 
In this algorithm, the min-priority queue $Q$ can be efficiently implemented by bucket sorting \cite{bz2003}, by which the total running time is optimized to linear $O(m+n)$.

\paragraph{Core Maintenance}
The problem of maintaining the core numbers for dynamic graphs $G$ is called core maintenance, when edges are inserted into and removed from $G$ continuously. The insertion and removal of vertices can be simulated as a sequence of edge insertions and removals. 
Hence, in this paper, we focus on maintaining the core numbers when an edge is inserted into or removed from a graph $G$.

\begin{definition}[Candidate Set $V^*$ and Searching Set $V^+$]
\label{df:set}
Given an undirected graph $G=(V,E)$, when an edge is inserted or removed, a candidate set of vertices, denoted as $V^*$, have to be computed so that the core numbers of all vertices in $V^*$ must be updated. In order to identify $V^*$, a minimal set of searching vertices, denoted as $V^+$, is traversed by repeatedly accessing their adjacent edges.
\end{definition}

Definition \ref{df:set} says that $V^*$ is identified by traversing all vertices in $V^+$, so that $V^*$ has to belong to $V^+$, denoted as $V^* \subseteq V^+$. Further, the vertices in $\vgr$ are traversed but not candidate vertices.
Efficient core maintenance algorithms should have a small ratio of $\vgrratio$ in order to minimize the cost of computing $V^*$. After $V^*$ is identified, the core number of vertices in $V^*$ can be updated accordingly.

\begin{figure}[htb]
\centering
\includegraphics[scale=0.4]{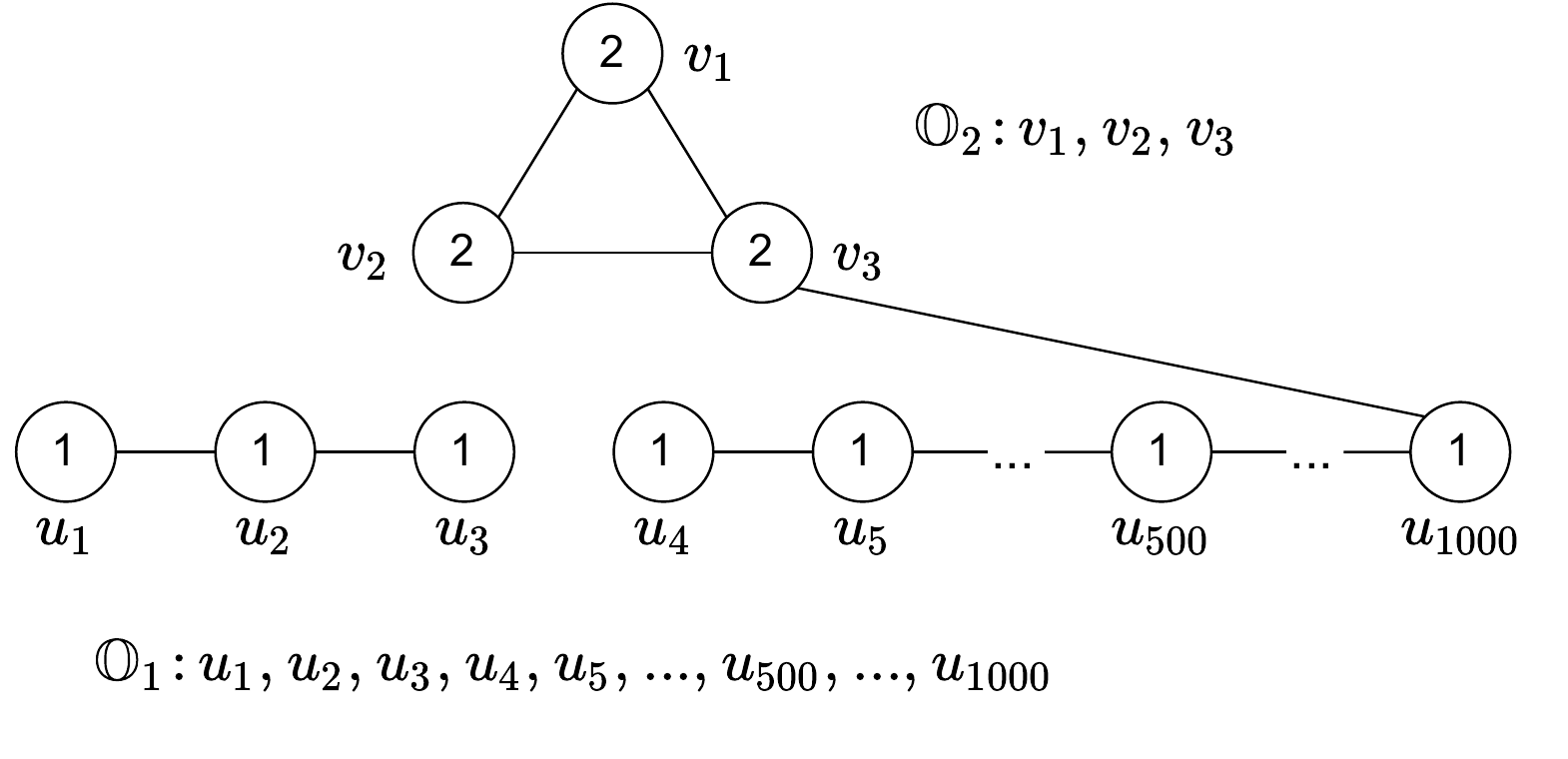}
\centering
\caption{A sample graph $G$ with $\od=\od_1\od_2$ in $k$-order.}
\label{fig:order}
\end{figure}

\begin{example}
Consider the graph $G$ in Figure \ref{fig:order}. The numbers inside the vertices are the core numbers. Three vertices, $v_1$ to $v_3$, have same core numbers of $2$; the other vertices, $u_1$ to $u_{1000}$, have same core numbers of $1$.
The whole graph $G$ is the $1$-core of since each vertex has a degree of at least $1$; the subgraph induced by $\{v_1, v_2, v_3\}$ is the $2$-core since each vertex in this subgraph has a degree of at least $2$.
After inserting an edge, for example $(u_1, u_{500})$, we observe that the core numbers of all vertices are not changed according to the peeling process. In this case, the candidate set $V^*=\emptyset$. However, the searching set $V^+$ is different for different edge insertion method, e.g., the order-based algorithm may have $V^+=\{u_1,u_2,u_3\}$ and the traversal algorithm traverse all vertices in $\texttt{sc}(u_1)$ with $V^+=\{u_1, u_2,\dots,u_{1000}\}$. 
\end{example}


We present two theorems given in \cite{li2013efficient,Saryuce2016,Zhang2017} which are useful to discuss the correctness of our insertion and removal algorithms. 

\begin{theorem}
\label{thr:core-number}
\rm{\cite{li2013efficient,Saryuce2016,Zhang2017}} After inserting an edge in or removing an edge from $G = (V, E)$, the core number of a vertex $u \in V^*$ increases or decreases by at most 1, respectively.
\end{theorem}


\begin{theorem}
\label{thr:core-range}
\rm{\cite{li2013efficient,Saryuce2016,Zhang2017}} Suppose an edge $(u, v)$ with $K =u.\core\leq v.\core$ is inserted to (resp. removed from) $G$. Suppose $V^*$ is non-empty. We have the following: 
(1) if $u.\core < v.\core$, then $u\in V^*$ and $V^* \subseteq \texttt{sc}(u)$ (as in Definition \ref{def:subcore}); 
(2) if $u.\core = v.\core$, then both vertices $u$ and $v$ are in $V^*$ (resp. at least one of $u$ and $v$ is in $V^*$) and $V^* \subseteq \texttt{sc}(u) \cup \texttt{sc}(v)$;
(3) the induced subgraph of $V^* \in G \cup \{(u, v)\}$ is connected.
\end{theorem}
 
Theorem \ref{thr:core-range} suggests that: (1) $V^*$ only includes the vertices $u\in V$ with $u.\core = K$; (2) $V^*$ can be searched in a small local region near the inserted or removed edge rather than in a whole graph. That is, to identify $V^*$, all vertices in $V^+$ are located in the subcores containing $u$ and $v$.

\paragraph{Order Data Structure}
The well-know \emph{Order Data Structure} \cite{dietz1987two, bender2002two} maintains a total order of subjects by following operations:
\begin{itemize}[noitemsep,topsep=0pt] 
    \item $\texttt{ORDER}(\od, x, y)$: determine if $x$ precedes $y$ in the total order~$\od$.
    \item $\texttt{INSERT}(\od, x, y)$: insert a new item $y$ after $x$ in the total order~$\od$.
    \item $\texttt{DELETE}(\od, x)$: remove an item $x$ from the total order $\od$.
\end{itemize}

In \cite{dietz1987two, bender2002two}, it is proved that all above three operations require worst-case $O(1)$ running time with linear space. 
The main idea is that each item $x$ in the total order is assigned a label to indicate the order. In this way, an $\texttt{ORDER}$ operation only requires $O(1)$ time for label comparisons and a $\texttt{DELETE}$ operation only require $O(1)$ time for directly removing one item without affecting the labels of other items.
Significantly, an $\texttt{INSERT}(\od, x, y)$ is complicated: 1) if there exists a valid label between two items $x$ and $x$'s successor, the new item $y$ can be inserted between them by assigning a new label, which requires $O(1)$ time; 2) or else, a \emph{relabel operation} is triggered to rebalance the labels for adjacent items, which requires $O(1)$ amortized running time.

In this work, our simplified order-based core maintenance algorithms are based on this Order Data Structure. Our time complexity analysis is based on the $O(1)$ time for the above three order operations.

\section{The Order-Based Algorithm}
In this section, we discuss the state-of-the-art order-based core maintenance approach in \cite{Zhang2017}. This algorithm is based on the $k$-order, which can be generated by the BZ algorithm for core decomposition \cite{bz2003} as in Algorithm \ref{alg:bz}. The $k$-order is defined as follows.


\begin{definition}[$k$-Order $\preceq$]
\rm{\cite{Zhang2017}}
Given a graph $G$, the $k$-order $\preceq$ is defined for any pairs of vertices $u$ and $v$ over the graph $G$ as follows: (1) when $u.\core < v.\core$, $u \preceq v$; (2) when $u.\core = v.\core$, $u \preceq v$ if $u$'s core number is determined before $v$'s by BZ algorithm (Algorithm \ref{alg:bz}, line \ref{alg:bz-core}). 
\end{definition}

A $k$-order $\preceq$ is an instance of all the possible vertex sequences produced by Algorithm~\ref{alg:bz}. For the $k$-order, transitivity holds, that is, $u\preceq v $ if $u\preceq w \land w\preceq v$. For each edge insertion and removal, the $k$-order will be maintained.  

Here, $\od_k$ denotes the sequence of vertices in $k$-order whose core numbers are $k$. 
A sequence $\od = \od_0\od_1\od_2 \cdots$ over $V(G)$ can be obtained, where $\od_i \preceq \od_j$ if $i < j$. It is clear that $\preceq$ is defined over the sequence of $\od = \od_0\od_1\od_2\cdots$. In other words, for all vertices in graph, the sequence $\od$ indicates the $k$-order $\preceq$. 


\begin{example}
Continually consider the graph $G$ in Figure \ref{fig:order}. The numbers inside the vertices are the core numbers. 
The $k$-order of $G$ is shown by $\od_1$ and $\od_2$, which is the order of core numbers determined by the BZ algorithm (Algorithm \ref{alg:bz} line~5); also, $\od_1$ is determined before $\od_2$, so that we have $\od_1\preceq \od_2$. 
\end{example}

\subsection{The Order-Based Insertion}
The key step for the insertion algorithm is to determine $V^*$.  
To do this, two degrees, $u.d^+$ and $u.d^*$, for each vertex $u\in V(G)$ are maintained in order to identify whether $u$ can be added into $V^*$ or not:
\begin{itemize} [noitemsep,topsep=0pt]
    \item remaining degree $u.d^+$: the number of the neighbors after vertex $u$ in $\od$ that can potentially support the increment of the current core number. 
    
    \item candidate degree $u.d^*$: the number of the neighbors before vertex $u$ in $\od$ that can potentially have their core number increased. 
\end{itemize}

Assume that an edge $(u, v)$ is inserted with $K=u.\core\leq v.\core$. 
The intuition behind the order-based insertion algorithm is as follows.
Starting from $u$, all affected vertices with the same core number $K$ (Theorem \ref{thr:core-range}) are traversed in $\od$. 
For each visited vertex $w\in V^+$, the value of $w.d^*+w.d^+$ is maximal as $w$ is visited by $k$-order. 
In this case, $w$ will be added into $V^*$ if $w.d^*+w.d^+ > K$; otherwise, $w$ is impossibly in $V^*$, which may repeatedly cause other vertices to be removed from $V^*$.
When all vertices with core number $K$ are traversed, this process terminates and $V^*$ is identified. 
Finally, the core numbers for all vertices in $V^*$ are updated by increasing by~1 (Theorem \ref{thr:core-number}). Obviously, for all vertices $u\in V$, the order $\od$ along with $u.d^+$ and $u.d^*$ must be maintained accordingly. 

Compared with the Traversal insertion algorithm \cite{Saryuce2016}, the benefit of traversing with $k$-order is that a large number of unnecessary vertices in $V^+\setminus V^*$ can be avoided. This is why the order-based insertion algorithm is generally more efficient. 


The order-based insertion algorithm is not easy to implement as it needs to traverse the vertices in $\od$ efficiently.
There are three cases.
First, given a pair of vertices $u, v\in \od_k$, the order-based insertion algorithm needs to efficiently test whether $u\preceq v$ or not.
For this, $\od_k$ is implemented as a double linked list associated with a data structure $\mathcal{A}_k$ which is a binary search tree and each tree node holds one vertex. 
For all $u, v \in \od_k$, we can test the order $u\preceq v$ in $O(\log |\od_k|)$ time by using $\mathcal{A}_k$. 
Second, the order-based insertion algorithm needs to efficiently ``jump'' over a large number of non-affected vertices that have $u.d^*=0$. To do this, $\od_k$ is also associated with a data structure $\mathcal B$, which is a min-heap.
Here, $\mathcal B$ supports finding a affected vertex $u$ with $u.d^*>0$ sequentially in $\od_k$ with $O(1)$ time; but it requires $O(\log|\od_k|)$ time to maintain the min-heap. 
Therefore, when maintaining $\od$, both $\mathcal A$ and $\mathcal B$ requires to updated accordingly, which requires worst-case $O(|V^+| \cdot \log|\od_k| + O(|V^*|)\log|\od_{k+1}|)$ time for removing $v\in V^*$ from $\od_k$ and then inserting $v\in V^*$ at the head of $\od_{k+1}$. 

As we can see, the $\mathcal A$ and $\mathcal B$ data structures are complicated, which complicates understanding and implementation. Additionally, the operations on $\mathcal A$ and $\mathcal B$ are time-consuming, especially when handling a data graph with a large sizes of $\od_k$ or $\od_{k+1}$.


\subsection{The Order-Based Removal}

The order-based removal algorithm adopts the same routine used in the traversal removal algorithm \cite{Saryuce2016} to compute $V^*$. This order-based removal algorithm is based on the \emph{max-core degree}.
\begin{definition} [max-core degree $\mcd$] \cite{Saryuce2016,Zhang2017}
Given a graph $G=(V,E)$, for each vertex $v\in V$, the max-core degree, $v.\mcd$, is the number of $v$'s neighbors $w$ such that $w.core\geq v.core$, defined as $v.\mcd = |\{w\in v.\adj: w.\core \geq v.\core \}|$.
\end{definition}

As discussed, the edge removal is much simpler than the edge insertion since edge removal is bounded for $V^*=V^+$.   
Assuming an edge $(u, v)$ is removed from the graph, both $u.\mcd$ and $v.\mcd$ are updated accordingly. 
This may repeatedly affect other adjacent vertices' $\mcd$. 
When the process terminates, all affected vertices $u$ that have $u.\mcd<u.\core$ can be added into $V^*$ and then their core numbers are off by $1$. 
Obviously, for all vertices $u\in V^*$, the sequence $\od$ along with $u.\mcd$ must be maintained accordingly.


Compared with the Traversal removal algorithm, the difference is that the order-based removal algorithm needs to maintain
$\od$ for all vertices in $V^*$. That is, all vertices in $V^*$ with core number $k$ are deleted from $\od_k$ and then appended to $\od_{k-1}$ in the corresponding $k$-order. Recall that two associated data structures, $\mathcal A$ and $\mathcal B$, are used for the order-based insertion algorithm. Both $\mathcal A$ and $\mathcal B$ must be updated accordingly, which requires worst-case $O(|V^*|\cdot (\log |O_k| + \log |O_{k-1}|))$ time for removing $v\in V^*$ from $\od_k$ and appending $v\in V^*$ at the tail of $\od_{k-1}$.
Analogously to the order-based insertion, the operations on $\mathcal A$ and $\mathcal B$ are time-consuming when handling a data graph with a large size of $\od_k$ or $\od_{k-1}$.

\section{The Simplified Order-Based Algorithm}
The main reason for the order-based algorithm being complicated and inefficient is that two data structures, $\mathcal A$ and $\mathcal B$, are used to maintain $\od$ in $k$-order for all vertices in a graph. 
In this section, we adopt the Order Data Structure \cite{dietz1987two, bender2002two} to maintain the $k$-order for all vertices. There are two benefits: one is that the $k$-order operations, such as inserting, deleting, and comparing the order of two vertices, can be optimized to $O(1)$ amortized running time; the other is that the original order-based method \cite{Zhang2017} can be simplified, which makes it easier to implement and to discuss the correctness.

Before introducing the new method, we propose a \emph{constructed Directed Acyclic Graph (DAG)} to simplify the statement of our algorithms. 
Given an undirected graph $G=(V,E)$ with $\od$ in $k$-order, each edge $(u, v)\in E(G)$ can be assigned a direction such that $u \preceq v$.
By doing this, a \emph{direct acyclic graph} (DAG) $\vec G=(V,\vec E)$ can be constructed where each edge $u\mapsto v\in \vec E(\vec G)$ satisfies $u\preceq v$. 
Of course, the $k$-order of $G$ is the \emph{topological order} of $\vec G$.
The \emph{post} of a vertex $v$ in $\vec G(V, \vec E)$ is all its successors (outgoing edges), defined by $u(\vec G).post=\{v \mid u\mapsto v \in \vec E (\vec G)\}$;
the \emph{pre} of a vertex $v$ in $\vec G(V, \vec E)$ is all its its predecessors (incoming edges), defined by $u(\vec G).pre=\{v \mid v\mapsto u \in \vec E (\vec G)\}$.
When the context is clear, we use $u.post$ instead of $u(\vec G).post$ and $u.pre$ instead of $u(\vec G).pre$.

In other words, the constructed DAG $\vec G=(V, \vec E)$ is equivalent to the undirected graph $G(V, E)$ by associating the direction for each edge in $k$-order. 
This newly defined constructed DAG $\vec G$ is convenient for describing our simplified order-based insertion algorithm.

\begin{lemma}
\label{lm:outdegree}
Given a constructed DAG $\vec G=(V,\vec E)$, for each vertex $v\in V$, the out-degree $|v.post|$ is not greater than the core number, $|v.\post| \leq v.\core$.
\end{lemma}
\begin{proof}
Since the topological order of $\vec G$ is the $k$-core of $G$, when removing the vertex $v$ by executing the BZ algorithm (Algorithm \ref{alg:bz}, line \ref{alg:bz-core}) all the vertices in $v.pre$ are already removed. In such a case, the out-degree of $v$ is its current degree. If there exist $|v.\post|>v.\core$ the value $v.\core$ should equal to $|v.\post|$, which leads to a contradiction. 
\end{proof}

If inserting an edge into a constructed DAG $\vec G$ does not violate Lemma \ref{lm:outdegree}, no maintenance operations are required. Otherwise, $\vec G$ has to be maintained to re-establish Lemma~\ref{lm:outdegree}. 

Table 1 summarizes the notations that will be frequently used when describing the algorithm. 
\begin{table}[htb]
\label{tb:notations}
\caption{Notations.}
\begin{tabular}{r|l}
\toprule
Notation &  Description\\
\midrule
$G=(V,E)$              &  an undirected graph\\
$\vec G = (V, \vec E)$ &   an constructed DAG by the $k$-order\\ 
$u\mapsto v \in \vec E (\vec G)$  & a directed edge in an constructed DAG \\
$\od=\od_0\od_1\dots\od_k$      & a sequence indicates the $k$-order $\preceq$\\
$u(\vec G).\degin$             &  the remaining in-degree of $u$\\ 
$u(\vec G).\degout$            &  the candidate out-degree of $u$\\ 
$u(\vec G).\post$        & the successors of $u$ in $\vec G$\\
$u(\vec G).\pre$         & the predecessor of $u$ in $\vec G$\\
$u.\mcd$              & the max-core degree of $u$\\
$u.\core$             & the core number of $u$\\
$V^*$                   & candidate set\\
$V^+$                   & searching set\\
$\Delta \vec G=(V, \Delta \vec E)$ & {an inserted graph}\\
\bottomrule
\end{tabular}
\end{table}

\subsection{The Simplified Order-Based Insertion}
\paragraph{Theory Background} With the concept of the constructed DAG $\vec G$, we can introduce our simplified insertion algorithm to maintain the core numbers after an edge is inserted to $\vec G$. For convenience, based on the constructed DAG $\vec G$, we first redefine the two concepts of candidate degree and remaining degree as in \cite{Zhang2017}. 

\begin{definition} [candidate in-degree]
Given a constructed DAG $\vec G(V,\vec E)$, the candidate in-degree $v.\degin$ is the total number of its predecessors located in $V^*$, denoted as $$v.\degin = |\{w\in v.pre: w\in V^*\}|$$
\end{definition}

\begin{definition} [remaining out-degree]
Given a constructed DAG $\vec G(V,\vec E)$, the remaining out-degree $v.\degout$ is the total number of its successors without the ones that are confirmed not in $V^*$, denoted as $$v.\degout = |\{w\in v.post: w\notin \vgr\}|$$
\end{definition}

In other words, assuming that $K=v.core$, the candidate in-degree $v.\degin$ counts the number of predecessors that are already in the new $(K+1)$-core; $v.\degout$ counts the number of successors that can be in the new $(K+1)$-core. Therefore, $v.\degin + v.\degout$ upper bounds the number of $v$'s neighbors in the new $(K+1)$-core.  

\begin{theorem} 
\label{th:potential-degree}
Given a constructed DAG $\vec G=(V,\vec E)$ by inserting an edge $u\mapsto v$ with $K = u.core \leq v.core$, the candidate set $V^*$ includes all possible vertices that satisfy: 1) their core numbers equal to $K$, and 2) their total numbers of candidate in-degree and remaining out-degree are greater than $K$, denoted as
$$\forall w \in V: w \in V^* \equiv (w.\core = K~\land~w.\degin + w.\degout > K)$$ 
\begin{proof}
According to Theorem~\ref{thr:core-number} and Theorem~\ref{thr:core-range}, for all vertices in $V^*$, we have 1) their core numbers equal to $K$, and 2) their core numbers will increase to $K+1$ and they can be added to new $(K+1)$-core. 
By the definition of $k$-core, for a vertex $v\in V^*$, $v$ must have at least $K+1$ adjacent vertices that can be in the new $(K+1)$-core. 
As $v.\degin+v.\degout$ is the number of $v$'s adjacent vertices that can be in the new $(K+1)$-core, we get $v.\degin+v.\degout > K$ for all vertices $v\in V^*$. 
\end{proof}
\end{theorem}

\begin{theorem}
\label{th:max-potential-degree}
Given a constructed DAG $\vec G=(V, \vec E)$ by inserting an edge $u \mapsto v$ with $u$ in $\od_K$, all affected vertices $w$ are after $u$ in $\od_K$. 
Starting from $u$, when $w$ is traversed in $\od_K$ and the $V^+$, $V^*$, $w.\degin$, $w.\degout$ are updated accordingly, each time the value of $w.\degin+w.\degout$ is maximal.
\begin{proof}
For all the vertices in the constructed DAG $\vec G$, $\od$ is the topological order in $\vec G$ according to the definition of $\vec G$. When traversing affected vertices $w$ in $G$ in such topological order, each time for $w$ all the affected predecessors must have been traversed, so that we get the value of $w.\degin$ is maximal; also, all the related successors are not yet traversed, so that the value of $w.\degout$ is also maximal.
Therefore, the total value of $w.\degin+w.\degout$ is maximal.
\end{proof}
\end{theorem}

In other words, when traversing the affected vertices $w$ in $\od$, $w.\degin+w.\degout$ is the upper-bound. That means, when traversing the vertices after $w$ in $\od$, $w.\degin+w.\degout$ only can be decrease as some of vertices can be removed from $V^*$. In this case, we can safely remove $w$ from $V^*$ if $w.\degin+w.\degout \leq K$, since $w$ is impossibly in $V^*$ according to Theorem~\ref{th:potential-degree}. This is the key idea behind the order-based insertion algorithm.

\paragraph{The Algorithm} 
Algorithm \ref{alg:insert} shows the detailed steps when inserting an edge $u\mapsto v$.
One issue is the implementation of traversing the vertices in $\od_k$. We propose to use a Min-Priority Queue combined with the Order Data Structure (line 4). The idea is as follow: 1) $\od_k$ is maintained by the Order Data Structure \cite{dietz1987two,bender2002two}, by which each vertex is assigned a label (an integer number) to indicate the order, and 2) all adjacent vertices are added into a Min-Priority Queue by using such labels as their keys. 
By doing this, we can dequeue a vertex from the Min-Priority Queue for each time to ``jumping'' over not-affected vertices efficiently. 
Further, three colors are used to indicate the different status for each vertex $v$ in a graph: 
\begin{itemize}
 \item  \texttt{white}: $v$ has initial status, $v\notin V^* \land v\notin V^+$.
 \item  \texttt{black}: $v$ is traversed and identified as a candidate vertex, $v\in V^* \land v\in V^+$.
 \item  \texttt{gray}: $v$ is traversed and identified impossible to be a candidate vertex, $v \notin V^* \land v \in V^+ \equiv v \in \vgr$
\end{itemize}

Before executing, we assume that for all vertices $v\in V(\vec G)$ their $\degout$ and $\degin$ are correctly maintained, that is $v.\degout = |v.\post|~\land~v.\degin = 0$. 
Initially, both $V^*$ and $V^+$ are empty (all vertices are \texttt{white}) and $K$ is initialized to $u.core$ since $u\preceq v$ for $u\mapsto v$ (line 1).
After inserting an edge $u\mapsto v$ with $u\preceq v$ in $\od$, we have $u.\degout$ increase by one (line 2). The algorithm will terminate if $u.\degout \leq u.\core$ as Lemma \ref{lm:outdegree} is satisfied (line 3). 
Otherwise, $u$ is added into the Min-Priority Queue $Q$ (line 4) for propagation (line \ref{alg:insert-while} to \ref{alg:insert-while-end}). 
For each $w$ removed from $Q$ (line 6), we check the value of $w.\degin+w.\degout$. That is, if $w.\degin+w.\degout > K$, vertex $w$ can be added to $V^*$ and may cause other vertices added in $V^*$, which is processed by the \texttt{Forward} procedure (line 7). Otherwise, $w$ cannot be added to $V^*$, which may cause some vertices to be removed from $V^*$ processed by the \texttt{Backward} procedure (line 8). 
Here, $w.\degin > 0$ means $w$ is affected, or else $w$ can be omitted since $w$ has no predecessors in $V^*$ (line 8). 
When $Q$ is empty, this process terminates and $V^*$ is obtained (line 5). 
At the ending phase, for all vertices $V^*$, their core numbers are increased by one (by Theorem \ref{thr:core-number}) and their $\degin$ are reset (line 9).
Finally, the $\od$ is maintained (line 10).

\begin{algorithm}[!htb]
\small
\SetAlgoNoEnd
\caption{EdgeInsert($\vec G, \od, u\mapsto v$)}
\label{alg:insert}
\DontPrintSemicolon
\SetKwInOut{Input}{input}\SetKwInOut{Output}{output}
\SetKwFunction{Min}{min}
\SetKwFunction{Forward}{Forward}
\SetKwFunction{Backward}{Backward}
\SetKwFunction{Ending}{Ending}
\SetKwFunction{DoDeg}{DoDeg}
\SetKwData{True}{true}
\SetKwData{False}{false}
\SetKwData{Black}{black}
\SetKwData{Dark}{dark}
\SetKwData{Gray}{gray}
\SetKwData{White}{white}
\SetKwData{Red}{red}
\SetKw{Break}{break}
\SetKwFunction{}{}
\SetKw{Continue}{continue}

\Input{A DAG $\vec G(V,\vec E)$; the corresponding $\od$; an edge $u\mapsto v$ to be inserted.}
\Output{An updated DAG $\vec G(V, \vec E)$; the updated $\od$.}

$V^*, V^+, K \gets \emptyset, \emptyset, u.\core$ \tcp*[r]{all vertices are \White}
insert $u\mapsto v$ into $\vec G$ with $u.\degout \gets u.\degout+1$\;
\lIf{$u.\degout \leq K$}{\textbf{return}}

$Q \gets $ a min-priority queue by $\od$; $Q.\enqueue (u)$\;
\While(){$Q \neq \emptyset$}{ \label{alg:insert-while}
    $w \gets Q.dequeue()$\;
    \lIf{$w.\degin + w.\degout > K$}{
        \Forward{$w, Q, K$}
    }
    \lElseIf{$w.\degin > 0$} {
        \Backward{$w, \od, K$}
    }
}\label{alg:insert-while-end}

\tcp*[l]{Ending Phase}
\lFor{$w\in V^*$}{$w.\core \gets K + 1$; $w.\degin \gets 0$}
\lFor{$w\in V^*$}{remove $w$ from $\od_K$ and insert $w$ at the beginning of $\od_{K+1}$ in $k$-order (the order $w$ added into $V^*$)}

\end{algorithm} 

The detailed steps of the \texttt{Forward} procedure are shown in Algorithm \ref{alg:insert-forward}. At first, $u$ is added to $V^*$ and $V^+$ (set from \texttt{white} to \texttt{black}) since $u$ has $u.\degin + u.\degout > K$ (line 1). 
Then, for each $u$'s successors $v$ whose core numbers equals to $K$ (by Theorem \ref{thr:core-range}), $v.\degin$ is increased by one (lines 2 and 3). In this case, $v$ is affected and has to be added into $Q$ for subsequent propagation (line 4).  

\begin{algorithm}[!htb]
\small
\SetAlgoNoEnd
\caption{Forward($u, Q, K$)}
\label{alg:insert-forward}
\DontPrintSemicolon
\SetKwInOut{Input}{input}\SetKwInOut{Output}{output}
\SetKwFunction{Min}{min}
\SetKwFunction{Forward}{Forward}
\SetKwFunction{Backward}{Backward}
\SetKwFunction{Ending}{Ending}
\SetKwFunction{DoDeg}{DoDeg}
\SetKwFunction{Delete}{{DELETE}}
\SetKwData{True}{true}
\SetKwData{False}{false}
\SetKwData{Black}{black}
\SetKwData{Dark}{dark}
\SetKwData{Gray}{gray}
\SetKwData{White}{white}
\SetKwData{Red}{red}
\SetKw{Break}{break}
\SetKw{With}{:}
\SetKwFunction{}{}
\SetKw{Continue}{continue}
    $V^* \gets V^* \cup \{u\}$; $V^+\gets V^+ \cup \{u\}$ \tcp*[r]{$u$ is \White to \Black}
    \For{$v \in u.post$ \With $v.\core = K$}{
        $v.\degin \gets v.\degin + 1$\;
        \lIf{$v \notin Q$}{ $Q.enqueue(v)$ }
    }
\end{algorithm} 

The detail steps of the \texttt{Backward} procedure are shown in Algorithm \ref{alg:insert-backward}. 
In the \texttt{DoPre}$(u)$ procedure, for all $u$'s predecessors $v$ that are located in $V^*$ (line 11), $v.\degout$ is decreased by one since $u$ is set to \texttt{gray} and cannot be added into $V^*$ any more (line 12); in this case, $v$ has to be added into $R$ for propagation if $v.\degin+v.\degout \leq K$(line 13).
Similarly, in the \texttt{DoPost}$(u)$ procedure, for all $u$'s successors $v$ that have $v.\degin > 0$ (line 15), $v.\degin$ is decreased by one (line 16) and added into $R$ for propagation if $v.\degin+v.\degout \leq K$ (lines 17 and 18).

The detailed steps of the \texttt{Backward} procedure are shown in Algorithm \ref{alg:insert-backward}. The queue $R$ is used for propagation (line 2). 
The \texttt{DoPre}$(u)$ procedure updates the graph when setting $u$ from \texttt{white} to \texttt{gray} or from \texttt{black} to \texttt{gray}, that is, for all $u$'s predecessors in $V^*$, all $\degout$ are off by $1$ and then added to $R$ for propagation, if its $\degin + \degout \leq K$ since they can not be in $V^*$ any more (lines 10 - 13). 
Similarly, the \texttt{DoPost} procedure updates the graph when setting $u$ from \texttt{black} to \texttt{gray}, that is, for all $u$'s successors with $\degin > 0$, all $\degin$ are off by $1$ and then added to $R$ for propagation if it is in $V^*$ and and its $\degin + degout \leq K$ (lines 14 - 18). 
Now, we explain the algorithm step by step. 
At first, $w$ is just added to $V^+$ (set from \texttt{white} to \texttt{gray}) since $w$ has $w.\degin+ w.\degout \leq K$ (line 1).
The queue $R$ is initialized as empty for propagation (line 2) and $w$ is propagated by the \texttt{DoPre} procedure. 
Of course, $w$'s $\degout$ and $\degin$ are updated (line 3) since all \texttt{black} vertices causing $w.\degin$ increased will be moved after $w$ in $\od$ eventually. 
All the vertices in $R$ are \texttt{black} waiting to be propagated (lines \ref{backward:while} to \ref{backward:while-end}).
For each $u\in R$, vertex $u$ is removed from $R$ (line 5) and removed from $V^*$, which sets $u$ from \texttt{black} to \texttt{gray} (line 6). 
This may require $\degin$ and $\degout$ of adjacent vertices to be updated, which is done by the procedures \texttt{DoPre} and \texttt{DoPost}, respectively (line 7). 
To maintain $\od_K$, $u$ is first removed from $\od_K$ and then inserted after $p$ in $\od_K$, where $p$ initially is $w$ or the previous moved vertices in $\od_K$ (line 8).
Of course, $u$'s $\degout$ and $\degin$ are updated (line 3) since all \texttt{black} vertices causing $u.\degin$ increased will be moved after $w$ in $\od$ eventually. 
This process is repeated until $R$ is empty (lines \ref{backward:while} to \ref{backward:while-end}).

\begin{algorithm}[!htb]
\small
\SetAlgoNoEnd
\caption{Backward($w, \od, K$)}
\label{alg:insert-backward}
\DontPrintSemicolon
\SetKwInOut{Input}{input}\SetKwInOut{Output}{output}
\SetKwFunction{Min}{min}
\SetKwFunction{Forward}{Forward}
\SetKwFunction{Backward}{Backward}
\SetKwFunction{Ending}{Ending}
\SetKwFunction{DoDeg}{DoDegree}
\SetKwFunction{DoPre}{DoPre}
\SetKwFunction{DoPost}{DoPost}
\SetKwFunction{Insert}{INSERT}
\SetKwFunction{Delete}{DELETE}
\SetKwData{Black}{black}
\SetKwData{Gray}{gray}
\SetKwData{White}{white}
\SetKwData{Red}{red}
\SetKw{Break}{break}
\SetKw{With}{$:$}
\SetKwFunction{}{}
\SetKw{Continue}{continue}

    $V^+\gets V^+ \cup \{w\}$; $p \gets w$\tcp*[r]{$w$ is \White to \Gray}

    $R \gets$ an empty queue;
    \DoPre{$w$}\;
    $w.\degout\gets w.\degout+w.\degin$; $w.\degin \gets 0$\;    
    
    \While(){$R \neq \emptyset$}{ \label{backward:while}
    $u \gets R.dequeue()$\;
    $V^* \gets V^* \setminus \{u\}$ \tcp*[r]{$u$ is \Black to \Gray}
    \DoPre{$u$}; \DoPost{$u$}\;
    \Delete($\od_K, u$); \Insert{$\od_K, p, u$}; $p \gets u$\;
    $u.\degout\gets u.\degout+u.\degin$; $u.\degin \gets 0$\;
    } \label{backward:while-end}

\medskip
\SetKwProg{myproc}{procedure}{ \DoPre{$u$}}{}
\myproc{}{
    \For{$v \in u.\pre$ \With $v\in V^*$}{
        $v.\degout \gets v.\degout - 1$\;
        \lIf{$v.\degin + v.\degout \leq K~\land~v\notin R$}{
                $R.enqueue(v)$
        }
        
    }
}
\medskip
\SetKwProg{myproc}{procedure}{ \DoPost{$u$}}{}
\myproc{}{
    \For{$v \in u.\post$ \With $v.\degin>0$}{
            $v.\degin \gets v.\degin - 1$\;
            \If{$v \in V^*~\land~v.\degin + v.\degout \leq K~\land~v\notin R$}{
                 $R.enqueue(v)$
            }
    }
}

\end{algorithm} 

\begin{figure}[htb]
\centering
\includegraphics[scale=0.4]{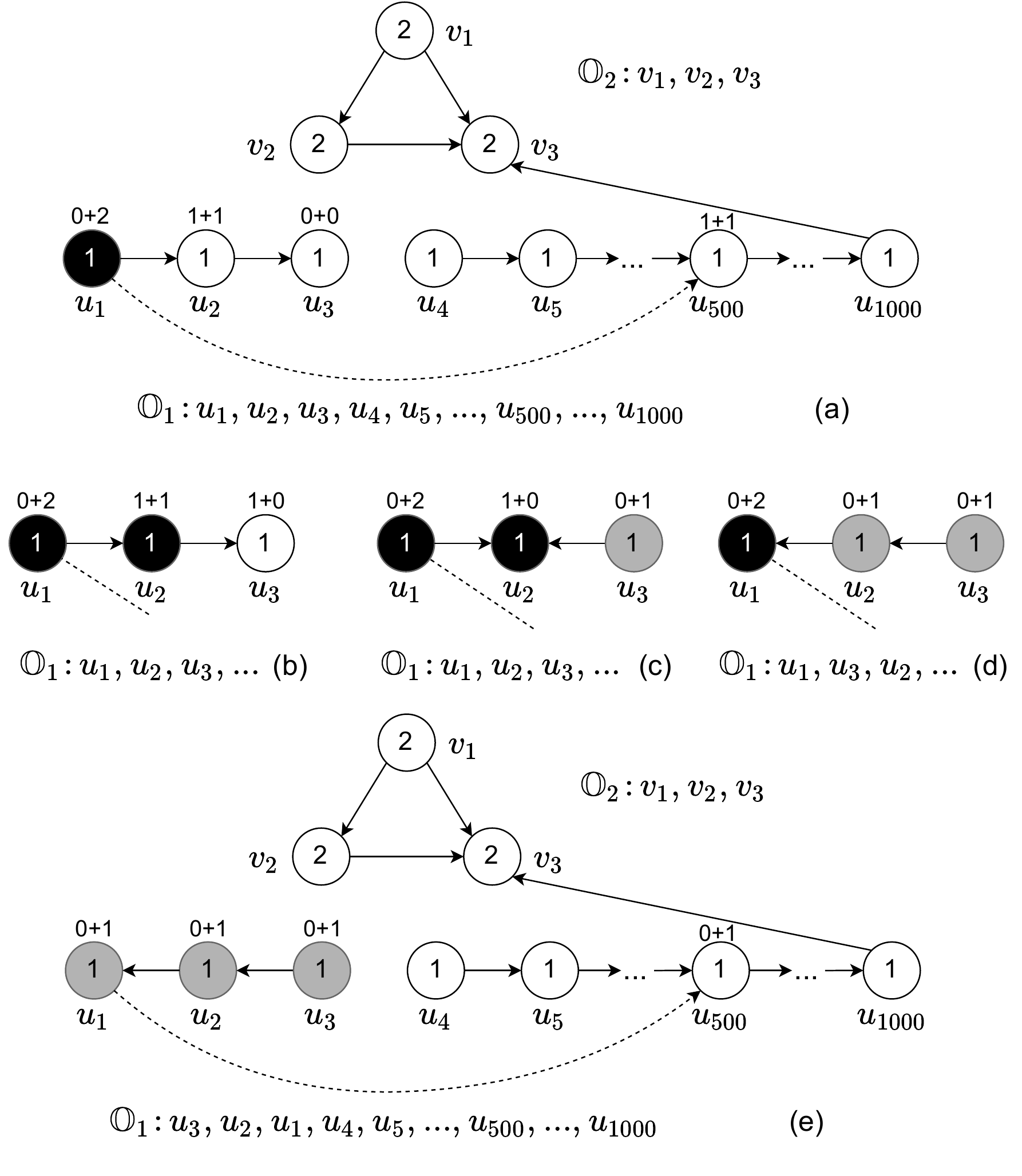}
\caption{Insert one edge $u_1\mapsto u_{500}$ to a constructed graph $\vec G$ obtained from Figure \ref{fig:order}.} 
\label{fig:insert}
\end{figure}

\begin{example}
Consider inserting an edge to a constructed graph in Figure~\ref{fig:insert} obtained from Figure \ref{fig:order}. The number inside the vertices are the core numbers, and the two numbers beside the vertices $u_1, u_2, u_3$ and $u_{500}$ are their $\degin+\degout$. Initially, we have the min-priority queue $Q=\emptyset$ and $K=1$.
In Figure \ref{fig:insert}(a), after inserting an edge $u_1\mapsto u_{500}$, we get $u_1.\degout =2~>~K$ and therefore $u_1$ is add to $Q$ as $Q=\{u_1\}$. 
We begin to propagate $Q$. 
First, in Figure~\ref{fig:insert}(a), $u_1$ is removed from $Q$ to do the \texttt{Forward} procedure since $u_1.\degout + u_1.\degin = 0 + 2~>~K$, by which $u_1$ is colored by \texttt{black}, all $u_1.\post$'s $\degin$ add by $1$, and all $u_1.\post$ are put into $Q$ as $Q=\{u_2, u_{500}\}$. 
Second, in Figure~\ref{fig:insert}(b), $u_2$ is removed from $Q$ to do the \texttt{Forward} procedure since $u_2.\degout+u_2.\degin = 1+1~>~K$, by which $u_2$ is colored by \texttt{black}, all $u_2.\post$'s $\degin$ add by $1$, and all $u_1.\post$ are added into $Q$ as $Q=\{u_3, u_{500}\}$. 
Third, in Figure~\ref{fig:insert}(c), however, $u_3$ is removed from $Q$ to do the \texttt{Backward} procedure since $u_3.\degin+ u_3.\degout = 1+0~\leq~K$, by which $u_3$ is colored by \texttt{gray} and we have $Q=\{u_{500}\}$. 

The \texttt{Backward} procedure continues.
In Figure~\ref{fig:insert}(d), we get $u_2.\degout$ off by $1$ and $u_2.\degin + u_2.\degout = 1+0~\leq~K$, so that $u_2$ is set to \texttt{gray}, by which $u_2$ is moved after $u_3$ in $\od_1$. 
In Figure~\ref{fig:insert}(e), we get $u_1.\degout$ off by $1$ and $u_1.\degin+u_1.\degout = 0+1~\leq~K$, so that $u_1$ is also set to \texttt{gray}, by which $u_1$ is moved after $u_3$ in $\od_1$; also, we get $u_{500}.\degin$ off by $1$ and the \texttt{Backward} procedure terminate.
Finally, we still need to check the last $u_{500}$ in $Q$, which can be safely omitted since its $\degin$ is $0$. 
In this simple example, we have $V^*=\emptyset \land V^+=\{u_1,u_2,u3\}$ and only $4$ vertices added to $Q$. A large number of vertices in $\od_1$, e.g. $u_4 \dots u_{1000}$, are avoid to be traversed.
\end{example}

\paragraph{Correctness.}
The key issue of the Algorithm \ref{alg:insert} is to identify the candidate set $V^*$. For the correctness, the algorithm has to be sound and complete.
The soundness implies that all vertices in $V^*$ are correctly identified,  
\[\fs sound(V^*)~\equiv~\forall v\in V: v\in V^* \Rightarrow  v.\degin+v.\degout >K~\land~v.\core = K\]
The completeness implies that all possible candidate vertices are added into $V^*$,     
\[\fs complete(V^*)~\equiv~\forall v \in V: v.\degin+v.\degout >K~\land~v.\core =K~\Rightarrow~v\in V^*\]
The algorithm has to ensure both the soundness and completeness
\[\fs sound(V^*)~\land complete(V^*), \] which is equivalent to
\[\fs \forall v\in V:  v\in V^*~\equiv~v.\degin+v.\degout > K~\land~v.\core = K.\]

To argue the soundness and completeness, we first define the vertices in $V(G)$ to have correct candidate in-degrees and remaining out-degrees as
\begin{equation*} \begin{split} \fs
in^*(V) & \equiv \forall v\in V: v.\degin = |\{w\in v.pre: w\in V^*\}| \\
out^+(V) & \equiv \forall v\in V: v.\degout = |\{w\in v.post: w\notin V^+ \setminus V^*\}|.
\end{split} \end{equation*}
We also define the sequence $\od$ for all vertices in $V$ are in $k$-order as
\[\fs \forall v_i\in V: \od(V) = (v_1, v_2, \dots , v_i) \Rightarrow v_1 \preceq v_2 \preceq \dots \preceq v_i \]

\begin{theorem} [soundness and completeness]
\label{th:insert-correct}
For any constructed graph $\vec G(V,\vec E)$, The while-loop in Algorithm \ref{alg:insert} (lines \ref{alg:insert-while} to \ref{alg:insert-while-end}) terminates with $sound(V^*)$ and $complete(V^*)$.
\end{theorem}
\begin{proof}
The invariant of the outer while-loop (lines \ref{alg:insert-while} to \ref{alg:insert-while-end} in Algorithm \ref{alg:insert}) is that all vertices in $V^*$ are sound, but adding a vertex to $V^*$ (\texttt{white} to \texttt{black}) may lead to its successors to be incomplete; for all vertices, their $\degin$ and $\degout$ counts are correctly maintained. 
All vertices in $Q$ have their core numbers as $K$ and their $\degin$ must be greater or equal to $0$, and all vertices $v\in V$ must be greater than $0$ if $v$ is located in $Q$; 
also, the $k$-order for all the vertices not in $V^*$ is correctly maintained:
\begin{equation*} \begin{split}
&sound(V^*)~\land~complete(V^*\setminus Q)~\land~in^*(V)~\land~out^+(V) \\
&~\land~(\forall v\in Q: v.\core = K \land v.\degin \geq 0)\\
&~\land~(\forall v\in V: v.\degin > 0 \Rightarrow v \in Q)\\
&~\land~\od(V\setminus V^*)\\
\end{split} \end{equation*}
The invariant initially holds as $V^* = \emptyset$ and for all vertices their $\degin$, $\degout$ and $k$-order are correctly initialized; also $u$ is first add to $Q$ for propagation only when $u.\core = K \land u.\degout > K \land u.\degin = 0$. We now argue that the while-loop preserves this invariant:
\begin{itemize} [leftmargin=*]
    \item[--] $sound(V^*)$ is preserved as $v\in V$ is added to $V^*$ only if $v.\degin+v.\degout > K$ by the \texttt{Forward} procedure; also, $v$ is safely removed from $V^*$ if $v.\degin+v.\degout \leq K$ by the \texttt{Backward} procedure according to Theorem \ref{th:max-potential-degree}. 
    
    \item[--] $complete(V^*\setminus Q)$ is preserved as all the affected vertices $v$, which may have $v.\degin + v.\degout > K$, are added to $Q$ by the \texttt{Forward} procedure for propagation.
    
    \item[--] $in^*(V)$ is preserved as each time when a vertex $v$ is added to $V^*$, all its successors' $\degin$ are increased by $1$ in the \texttt{Forward} procedure; also each time when a vertex $v$ cannot be added to $V^*$, the $\od$ may change \texttt{Backward} procedure.
    
    \item[--] $out^*(V)$ is preserved as each time when a vertex $v$ cannot be added to $V^*$, the $\od$ may change and the corresponding $\degout$ are correctly maintained by the \texttt{Backward} procedure.
    
    \item[--] $(\forall v\in Q: v.\core = K \land v.\degin \geq 0)$ is preserved as in the \texttt{Forward} procedure, the vertices $v$ are added in $Q$ only if $v.core = K$ with $v.\degin$ add by $1$; but $v.\degin$ may be reduced to $0$ in the \texttt{Backward} procedure when some vertices cannot in $V^*$.
    
    \item[--] $(\forall v\in V: v.\degin > 0 \Rightarrow v \in Q)$ is preserved as $v$ can be added in $Q$ only after adding $v.\degin$ by $1$ in the \texttt{Forward} procedure.
    
    \item[--] $\od(V\setminus V^*)$ is preserved as the $k$-order of all vertices $v\in \vgr$ is correctly maintained by the \texttt{Backward} procedure and the $k$-order of all the other vertices $v\in V\setminus V^+$ is not affected.
\end{itemize}

We also have to argue the invariant of the inner while-loop in the \texttt{Backward} procedure (lines \ref{backward:while} to \ref{backward:while-end} in Algorithm \ref{alg:insert-backward}).
The additional invariant is that all vertices in $R$ has to be located in $V^*$ but not sound as their $\degin +\degout \leq K$:
\begin{equation*} \begin{split}
&sound(V^*\setminus R)~\land~complete(V^* \setminus {Q})~\land~in^*(V)~\land~out^+(V)\\
&~\land~{(\forall v\in R: v.core = K \land v\in V^* \land v.\degin + v.\degout \leq K)}\\
&~\land~(\forall v\in Q: v.core = K \land v.\degin \geq 0)\\
&~\land~(\forall v\in V: v.\degin > 0 \Rightarrow v \in Q)\\
&~\land~{\od(V\setminus V^*)}
\end{split} \end{equation*}

The invariant initially holds as for $w$, all its predecessors' $\degout$ are off by 1 and added in $R$ if their $\degin + \degout \leq K$ since $w$ is identified in $\vgr$ (\texttt{gray}). 
We have $w$ $\preceq$ all vertices in $V^*$ in $\od$, denoted as $w\preceq V^*$, as 1) $v$ can be moved to the head of $\od_{K+1}$ and $v.\core$ is add by $1$ if $v$ is still in $V^*$ when the outer while-loop terminated, and 2) $v$ is removed from $V^*$ and moved after $w$ in $O_K$. In this case, $w.\degout$ and $w.\degin$ are can be correctly updated to $(w.\degout + w.\degin)$ and $0$, respectively. 
We now argue that the while-loop preserves this invariant:
\begin{itemize} [leftmargin=*]
    \item[--] $sound(V^*\setminus R)$ is preserved as all $v\in V^*$ are added to $R$ if $v.\degin + v.\degout \leq K$. 
    
    \item[--] $in^*(V)$ is preserved as each time for a vertex $u\in R$ setting from \texttt{black} to \texttt{gray}, for all its affected successor, which have $\degin > 0$, their $\degin$ are off by 1; also, $u.\degin$ is set to $0$ when setting from \texttt{black} to \texttt{gray} since $u\preceq$ all vertices in $V^*$ in the changed $\od$.
    
    \item[--] $out^*(V)$ is preserved as each time for a vertex $u\in R$ setting from \texttt{black} to \texttt{gray}, for all its affected predecessor, which are in $V^*$, their $\degout$ are off by 1; also, $u.\degout$ is set to $u.\degout+u.\degin$ since $u\preceq$ all vertices in $V^*$ in the changed $\od$. 
    
    \item[--] ${(\forall v\in R: v.core = K \land v\in V^* \land v.\degin + v.\degout \leq K)}$ is preserved as each time for a vertex $v\in V^*$, $v$ is checked when $v.\degin$ or $v.\degout$ is off by 1, and $v$ is added to $R$ if $v.\degin + v.\degout \leq K$. 
    
    
    \item[--] $\od(V\setminus V^*)$ is preserved as each vertex $v$ that removed from $V^*$ by pealing are moved following $p$ in $O_K$, where $p$ is $w$ or the previous vertex removed from $V^*$. 
\end{itemize}

At the termination of the inner while-loop, we get $R=\emptyset$. At the termination of the outer while-loop, we get $Q=\emptyset$. The postcondition of the outer while-loop is $sound(V^*)\land complete(V^*)$. 
\end{proof}

At the ending phase of Algorithm \ref{alg:insert}, the core numbers of all vertices in $V^*$ are add by 1, and $\od$ in $k$-order is maintained. On the termination of Algorithm \ref{alg:insert}, the core numbers are correctly maintained and also $\od(V) \land in^*(V) \land out^+(V)$, which provide correct initial state for the next edge insertion.

\paragraph{Complexity}
\begin{theorem} The time complexity of the simplified order-based insertion algorithm is $O(|E^+| \cdot \log |E^{+}|)$ in the worst case, where $|E^+|$ is the number of adjacent edges for all vertices in $V^+$ defined as $|E^+| = \sum_{v\in V^{+}}^{} v.\Deg$.
\begin{proof}
As the definition of $V^+$, it includes all traversed vertices to identify $V^*$. 
In the \texttt{Forward} procedure, the vertices in $V^+$ are traversed at most once, so do in the \texttt{Backward} procedure, which requires worst-case $O(|E^+|)$ time. In the while-loop (Algorithm \ref{alg:insert} lines 5 - 10), the min-priority queue $Q$ includes at most $|E^+|$ vertices since each related edge of vertices in $V^+$ is added into $Q$ at most once. The min-priority queue can be implemented by min-heap, which requires worst-case $O(|E^+| \cdot \log |E^+|)$ time to dequeue all the values. All the vertices in $\od$ are maintained with Order Data Structure, so that manipulating the order of one vertex requires amortized $O(1)$ time; there are totally at most $|V^+|$ vertices whose order are manipulated, which requires worst-case $O(|V^+|)$ amortized time. Therefore, the total worst-case time complexity is $O(|E^+| + |E^+|\cdot \log |E^+| + |V^+|) = O(|E^+|\cdot \log |E^+|)$.
\end{proof}
\end{theorem}

\begin{theorem} The space complexity of the simplified order-based insertion algorithm is $O(n)$ in the worst-case.
\iftrue
\begin{proof}
Each vertex $v$ is assigned three counters that are $v.\core$, $v.\degin$ and $v.\degout$, which requires $O(3n)$ space.
Both $Q$ and $R$ have at most $n$ vertices, respectively, which require worst-case $O(2n)$ space together. 
Two arrays are required for $V^+$ and $V^*$, which requires worst-case $O(2n)$ space. 
All vertices in $\od$ are maintained by Order Data Structure. For this, all vertices are linked by double linked lists, which requires $O(2n)$ space; also, vertices are assigned labels (typically 64 bits integer) to indicate the order, which requires $O(2n)$ space.
Therefore, the total worst-case space complexity is 
$O(3n+2n+2n+2n+2n) = O(n)$.
\end{proof}
\else 
\begin{proof}
Each vertex $v$ is assigned three counters that are $v.\core$, $v.\degin$ and $v.\degout$.
Each vertex $v$ is assigned two flags to indicate $v$ is in $Q$ and $R$ or not, respectively.
Both $Q$ and $R$ have at most $n$ vertices, respectively. 
Two arrays are required for $V^+$ and $V^*$. 
Each vertex $v$ is assigned a color that is \texttt{white, gray} or \texttt{black}. 
All vertices in $\od$ are maintained by Order Data Structure. For this, all vertices are linked by double linked lists,
and vertices are assigned labels (typically 64 bits integer) to indicate the order.
Therefore, the total worst-case space complexity is $O(n)$.
\end{proof}
\fi
\end{theorem}

\subsection{The Simplified Order-Based Removal}
Our simplified order-based removal Algorithm is mostly the same as the original order-based removal Algorithm in \cite{Zhang2017,Saryuce2016}, so that the details are omitted in this section.
The only difference is that our simplified order-based removal algorithm adopts the Order Data Structure to maintain $\od$, instead of the complicated $\mathcal A$ and $\mathcal B$ data structures \cite{Zhang2017}. 
In this case, the worst-case time complexity can be improved as the Order Data Structure only requires amortized $O(1)$ time for each order operation.

\paragraph{Complexities} 
\begin{theorem} The time complexity of the simplified order-based removal algorithm is $O(\DEG(G) + |E^*|)$ in the worst case, where $|E^*| = \sum_{ w\in \it V^*}  {w.d}$. 
\begin{proof}
Typically, the data graph $G$ is stored by adjacent lists. For removing an edge $(u, v)$, all edges of the vertex $u$ and $v$ are sequentially traversed, which require at most $O(\DEG(G))$ time.
We know that $V^+$ includes all traversed vertices to identify the candidate set $V^*$ and $V^*=V^+$ in this algorithm. The vertices in $V^*$ are traversed at most once, which requires worst-case $O(|E^*|)$ time. All vertices in $V^*$ are removed from the $\od_{K}$ and appended to $\od_{K-1}$ in $k$-order, which require $O(|V^*|)$ time as each insert or remove operation only needs amortized $O(1)$ time by the Order Data Structure. 
Since it is possible that $\DEG(G) > |E^*|$ in some cases like $V^*=\emptyset$, the total worst-case time complexity is $O(\DEG(G) + |E^*|+|V^*|) = O(\DEG(G) + |E^*|)$.
\end{proof}
\end{theorem}

\begin{theorem} The space complexity of the simplified order-based removal algorithm is $O(n)$ in the worst-case.
\begin{proof}
For each vertex $v$ in the graph, $v.\mcd$ is used to identify the $V^*$, which requires $O(n)$ space. All vertices in $\od$ are maintained by Order Data Structure, which requires $O(4n)$ space. A queue is used for the propagation, which require worst-case $O(n)$ space. One array is required for $V^*$, which requires $O(n)$ space. Therefore, the total worst-case space is $O(n+4n+n+n)=O(n)$. 
\end{proof}
\end{theorem}

\iftrue 

\section{The Simplified Order-Based Batch Insertion}
In practice, it is common that a great number of edges are inserted into a graph together.
If multiple edges are inserted one by one, the vertices in $\vgr$ may be repeatedly traversed. 
Instead of inserting one by one, we can handle the edge insertion in batch.
In this section, we extend our simplified order-based unit insertion algorithm to batch insertion.

Let $\Delta G=(V, \Delta E)$ be an inserted graph to a constructed DAG $\vec G$. That is, $\Delta E(\Delta G)$ contains a batch of edges that will be inserted to $\vec G$. Each edge $u\mapsto v \in \Delta E$ satisfies $u\preceq v$ in the $k$-order of $\vec G$.  



\begin{theorem}
\label{th:batch-insert}
After inserted a graph $\Delta G=(V, \Delta E)$ to constructed DAG $\vec G=(V, \vec E)$, the core number of a vertex $v\in V(\vec G)$ increases by at most $1$ if $v$ satisfies $|v.post| \leq v.\core + 1$.
\begin{proof}
For each $v\in V(\vec G)$, Lemma \ref{lm:outdegree} proves that the out-degree of $v$ satisfies $|v.\post|\leq v.\core$. Analogies, when inserting $\Delta G$ into $\vec G$ with $|v.\post| \leq v.\core+1$, the core number can be increased by at most $1$, as after inserting the new graph has to satisfy $v.d_{out} \leq v.\core$ for all vertices $v\in V(\vec G)$. 
\end{proof}
\end{theorem}

Theorem \ref{th:batch-insert} suggests that each round we can insert multiples edges $u\mapsto v\in \Delta E$ into $\vec G$ only if $|u(\vec G).\post|\leq u(\vec G).core+1$; otherwise, $u\mapsto v$ has to be inserted in next round until all edges are inserted. In the worst-case, there are $\DEG(\Delta G)$ round required if each edges $u\mapsto v\in \Delta E$ satisfy $u(\vec G).\degout = u(\vec G).\core$. 



\paragraph{The Algorithm} Algorithm \ref{alg:batch-insert} shows the detailed steps.
A batch of edges $u\mapsto v\in \Delta E$ can be inserted into $\vec G$ only if $u.\degout \leq u.\core$ (lines 3 and 4).
When $u.\degout = u.\core + 1$, we can put $u$ into the Min-Priority Queue $Q$ for propagation (line 5). Of course, the inserted edges are removed from $\Delta G$ (line 6). 
After all possible edges are inserted, the propagation is the same as in lines 5 - 10 of Algorithm \ref{alg:insert} (line 7), where $K$ is the core numbers of local $k$-subcore with $K=u.\core \leq v.\core$ for an inserted edge $u\mapsto v$. 
This process repeatedly continues until the $\Delta G$ becomes empty (line 1).

\begin{algorithm}[htb]
\small
\SetAlgoNoEnd
\caption{BatchEdgeInsert($\vec G, \od, \Delta G$)}
\label{alg:batch-insert}
\DontPrintSemicolon
\SetKwInOut{Input}{input}\SetKwInOut{Output}{output}
\SetKwFunction{Min}{min}
\SetKwFunction{Forward}{Forward}
\SetKwFunction{Backward}{Backward}
\SetKwFunction{Ending}{Ending}
\SetKwFunction{DoDeg}{DoDeg}
\SetKwData{True}{true}
\SetKwData{False}{false}
\SetKwData{Black}{black}
\SetKwData{Dark}{dark}
\SetKwData{Gray}{gray}
\SetKwData{White}{white}
\SetKwData{Red}{red}
\SetKw{Break}{break}
\SetKw{With}{$:$}
\SetKw{WITH}{with}
\SetKw{IN}{in}
\SetKwFunction{}{}
\SetKw{Continue}{continue}
\SetKw{Break}{break}
\Input{A constructed DAG $\vec G=(V,\vec E)$; the corresponding $\od$; An inserted graph $\Delta G=(V, \Delta E)$. }
\Output{A updated DAG $\vec G(V, \vec E)$; A updated $\od$.}

    

        

\While{$\Delta G \neq \emptyset$}{
    $V^*, V^+, Q \gets \emptyset, \emptyset$, a min-priority queue by $\od$\;
    
    
    \For{$u \mapsto v \in \Delta E(\Delta G)$ \With $u.\degout \leq u.\core$}{
        insert $u\mapsto v$ into $\vec G$ with $u.\degout \gets u.\degout+1$\;
        
        \lIf{$u.\degout = u.\core+1$}{$Q.enqueu(u)$}
        remove $u\mapsto v$ from $\Delta G$\;
        
    }
    same code as lines 5 - 10 in Algorithm \ref{alg:insert} with $K$ as the core number of local subcore\;
}

\end{algorithm} 

\begin{figure}[htbp]
\centering
\includegraphics[scale=0.4]{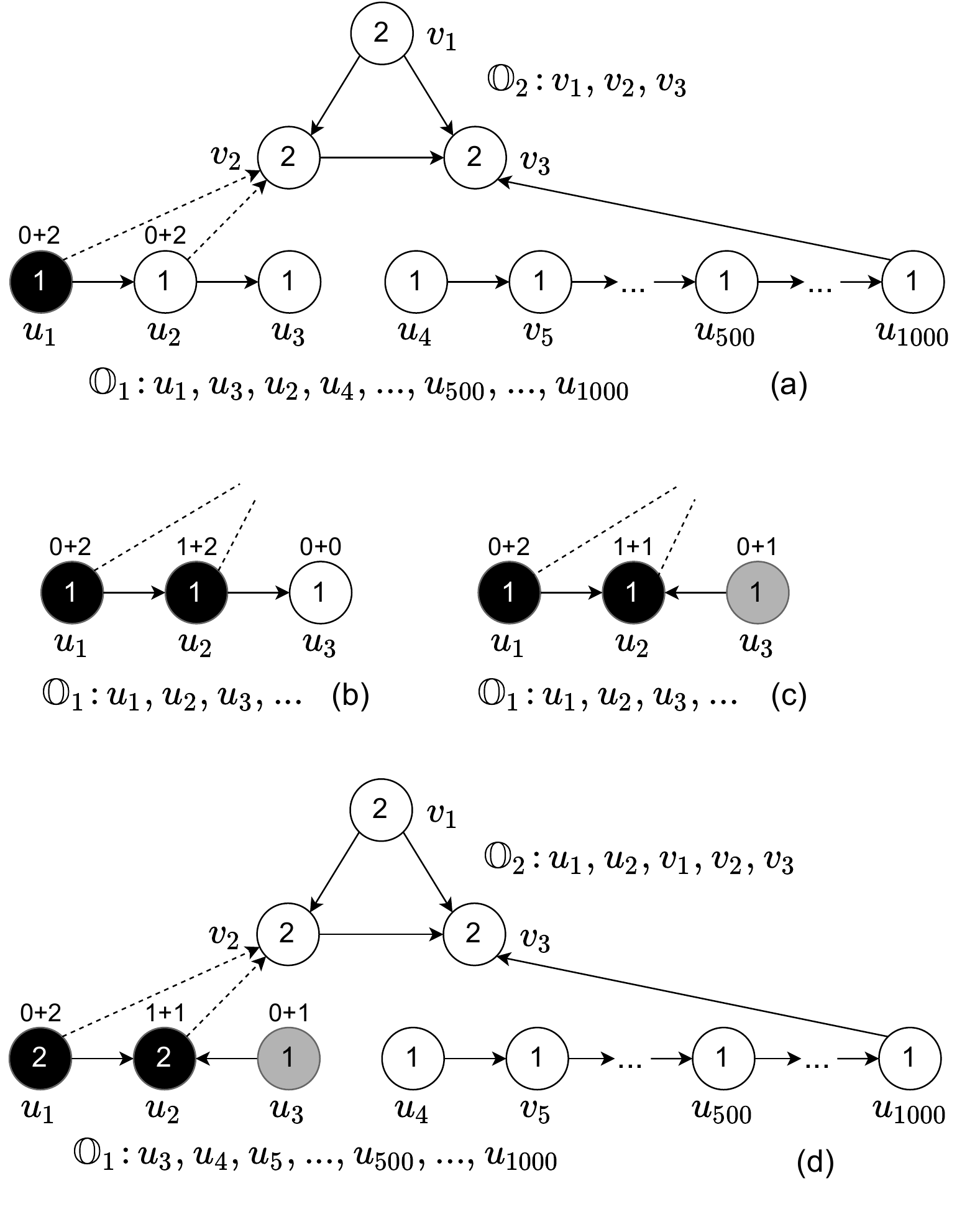}
\caption{Insert a batch of two edges $u_1\mapsto v_{2}$ and $u_2\mapsto v_2$ to a constructed graph $\vec G$ obtained from Figure \ref{fig:order}.}
\label{fig:batch-insert}
\end{figure}



\begin{example}
Consider inserting two edges in the constructed graph in Figure \ref{fig:batch-insert}. Initially, the Min-Priority Queue $Q$ is empty and $K$ is the core number of the corresponding $k$-subcore.
In Figure~\ref{fig:batch-insert}(a), after inserting two edges, $u_1 \mapsto v_2$ and $u_2 \mapsto v_2$, we get $u_1.\degout = K + 1 = 2$ and $u_2.\degout = K + 1 = 2$, so that these two edges can be inserted in batch and we put $u_1$ and $u_2$ in $Q$ as $Q=\{u_1, u_2\}$. We begin to propagate $Q$. 
First, in Figure~\ref{fig:batch-insert}(a), $u_1$ is removed from $Q$ to do the \texttt{Forward} procedure since $u_1.\degin + u_1.\degout = 0+2 > K = 1$, by which $u_1$ is colored by \texttt{black}; within subcore \texttt{sc}$(u_1)$, all $u_1.\post$'s $\degin$ are added by $1$, and all $u_1.\post$ are put in $Q$ as $Q=\{u_2\}$. 
Second, in Figure~\ref{fig:batch-insert}(b), $u_2$ is removed from $Q$ to do the \texttt{Forward} procedure since $u_2.\degin + u_2.\degout = 0 + 2 > K = 1$, by which $u_2$ is colored by \texttt{black}; within subcore \texttt{sc}$(u_2)$, all $u_2.\post$'s $\degin$ are added by $1$, and all $u_2.\post$ are put in $Q$ as $Q=\{u_3\}$.
Third, in Figure \ref{fig:batch-insert}(c), $u_3$ is removed from $Q$ to do the \texttt{Backward} procedure since $u_3.\degin + u_3.\degout = 1+0 \leq K = 1$, by which $u_3$ is colored by \texttt{black} and $u_2.\degout$ off by 1; however, since $u_2.\degin+u_2.\degout=1+1 > K = 1$, we have $u_2$ still \texttt{black} and the \texttt{Backward} procedure terminates.
Finally, in Figure \ref{fig:batch-insert}(d), two \texttt{black} vertices, $u_1$ and $u_2$, have increased core numbers as $2$; then, they are removed from $\od_1$ and inserted before the head of $\od_2$ to maintain the $k$-order. 

In this example, we have $V^*=\{u_1, u_2\}\land V^+=\{u_1,u_2,u_3\}$ by batch inserting two edges together. If we insert $u_1\mapsto v_2$ first and then insert $u_2\mapsto v_2$ second, the final $V^*$ is same; but $V^+$ is $\{u_1, u_2, u_3\}$ and $\{u_2, u_3\}$ for two inserted edges, respectively. In this case, both $u_2$ and $u_3$ are repeatedly traversed, which can be avoided by batch insertion.
\end{example}

\paragraph{Correctness} For each round of the while-loop (lines 2 - 7), the correctness argument is totally the same as the single edge insertion in Algorithm \ref{alg:insert}. 

\paragraph{Complexities} The total worst-case running time of line 7 is the same as in Algorithm \ref{alg:insert}, which is $O(|E^+|\cdot \log|E^+|)$. Typically, for $Q$, the running time of enqueue and dequeue are larger than in Algorithm \ref{alg:insert} since each time numerous vertices can be initially added into $Q$ for propagation (line 5).
The outer while-loop (line 1) run at most $\Delta E$ rounds, so that $\Delta E$ is checked at most $O(\DEG(\Delta G)\cdot |\Delta E|)$ round as $\od$ can be changed and thus the directions of edges in $\Delta E$ can be changed. 
Typically, the majority of edges can be inserted in the first round of while-loop.
Therefore, the time complexity of Algorithm \ref{alg:batch-insert} is $O(|E^+|\cdot \log|E^+| + \DEG(\Delta G)\cdot |\Delta E|)$ in the worst case.

The space complexity of Algorithm \ref{alg:batch-insert} is the same as Algorithm \ref{alg:insert}.

\fi

\section{Related Work}

\paragraph{Core Decomposition} 
In \cite{cheng2011efficient}, Cheng et al. propose an external memory algorithm, so-called EMcore, which runs in a top-down manner such that the whole graph does not have to be loaded into memory. In \cite{wen2016efficient}, Wen et al. provide a semi-external algorithm, which requires $O(n)$ size memory to maintain the information of vertices. In \cite{khaouid2015k}, Khaouid et al. investigate the core decomposition in a single PC over large graphs by using  \texttt{GraphChi} and \texttt{WebGraph} models. In \cite{montresor2012distributed}, Montresoret et al. consider the core decomposition in a distributed system. 
In addition, Parallel computation of core decomposition in multi-core processors is first investigated in \cite{dasari2014park}, where the ParK algorithm was proposed. Based on the main idea of ParK, a more scalable PKC algorithm has been reported in \cite{kabir2017parallel}.

\paragraph{Core Maintenance} 
In~\cite{li2013efficient}, an algorithm that similar to traversal algorithm~\cite{Saryuce2016} is given, but this solution has quadratic time complexity.
In \cite{wen2016efficient}, a semi-external algorithm for core maintenance is proposed in order to reduce the I/O cost, but this method is not optimized for CUP time. 
In~\cite{wang2017parallel,Jin2018}, parallel approaches for core maintenance is proposed for both edge insertion and removal. 
There exists some research based on core maintenance. In~\cite{yu2021querying}, the authors study computing all $k$-cores in the graph snapshot over the time window. In~\cite{lin2021hierarchical}, the authors explore the hierarchy core maintenance.

\section{Experiments}
In this section, we conduct experimental studies using 12 real and synthetic graphs and report the performance of our algorithm by comparing with the original order-based method:
\begin{itemize} [leftmargin=*]
    \item [--] The order-based algorithm \cite{Zhang2017} with unit edge insertion (\texttt{I}) and edge removal (\texttt{R}); Before running, we execute the initialization (\texttt{Init}) step. 
    \item [--] Our simplified order-based with unit edge insertion (\texttt{OurI}) and edge removal (\texttt{OurR}); Before running, we execute the initialization (\texttt{OurInit}) step.
    \item [--] Our simplified order-based batch edge insertion (\texttt{OurBI}).
\end{itemize}
The experiments are performed on a desktop computer with an Intel CPU (4 cores, 8 hyperthreads, 8 MB of last-level cache) and 16 GB main memory. 
The machine runs the Ubuntu Linux (18.04) operating system. All tested algorithms are implemented in C++ and compiled with g++ version 7.3.0 with the -O3 option. All implementations and results are available at github\footnote{\url{https://github.com/Itisben/SimplifiedCoreMaint.git}}. 

\paragraph{Tested Graphs}
We evaluate the performance of different methods over a variety of real-world and synthetic graphs, which are shown in Table \ref{tb:graph}. For simplicity, directed graphs are converted to undirected ones in our testing; all of the self-loops and repeated edges are removed, that is, a vertex can not connect to itself and each pair of vertices can connect with at most one edge. 
The \emph{livej}, \emph{patent}, \emph{wiki-talk}, and \emph{roadNet-CA} graphs are obtained from SNAP\footnote{\url{http://snap.stanford.edu/data/index.html}}. 
The \emph{dbpedia}, \emph{baidu}, \emph{pokec} and \emph{wiki-talk-en} \emph{wiki-links-en} graphs are collected from the KONECT\footnote{\url{http://konect.cc/networks/}} project. 
The \emph{ER}, \emph{BA}, and \emph{RMAT} graphs are synthetic graphs; they are generated by the SNAP\footnote{\url{http://snap.stanford.edu/snappy/doc/reference/generators.html}} system using Erd\"{o}s-R\'enyi, Barabasi-Albert, and the R-MAT graph models, respectively. For these generated graphs, the average degree is fixed to 8 by choosing 1,000,000 vertices and 8,000,000 edges.

\begin{table}[!htb]
\small
\caption{Tested real and synthetic graphs.}
\label{tb:graph}
\begin{tabular}{l|rrrr}
\toprule
Graph & $n$ & $m$ & AvgDeg & Max $k$ \\ 
\midrule
livej & 4,847,571 & 68,993,773 & 14.23 & 372 \\
patent & 6,009,555 & 16,518,948 & 2.75 & 64 \\
wikitalk & 2,394,385 & 5,021,410 & 2.10 & 131 \\
roadNet-CA & 1,971,281 & 5,533,214 & 2.81 & 3 \\ \hline
dbpedia & 3,966,925 & 13,820,853 & 3.48 & 20 \\
baidu & 2,141,301 & 17,794,839 & 8.31 & 78 \\
pokec & 1,632,804 & 30,622,564 & 18.75 & 47 \\
wiki-talk-en & 2,987,536 & 24,981,163 & 8.36 & 210 \\
wiki-links-en & 5,710,993 & 130,160,392 & 22.79 & 821 \\ \hline
ER & 1,000,000 & 8,000,000 & 8.00 & 11 \\
BA & 1,000,000 & 8,000,000 & 8.00 & 8 \\
RMAT & 1,000,000 & 8,000,000 & 8.00 & 237 \\
\bottomrule
\end{tabular}
\end{table}

In Table \ref{tb:graph}, we can see all graphs have millions of edges, their average degrees are ranged from 2.1 to 22.8, and their maximal core numbers range from 3 to 821. For each tested graph, the distribution of core numbers for all the vertices is shown in Figure \ref{fig:graph}, where the x-axis is the core numbers, and the y-axis is the size of vertices. 
For most graphs, many vertices have small core numbers, and few have large core numbers. 
Specifically, \emph{wiki-link-en} has the maximum core numbers up to $821$, so that for most of its $\od_k$ the sizes are around $1000$; \emph{BA} has a single core number as $8$ so that all vertices are in the single $\od_k$.
Since all vertices with core number $k$ in $\od_k$ are maintained in $k$-order, the size of $\od_k$ is related to the performance of different methods.

\begin{figure}[!htb]
\centering
\includegraphics[scale=0.4]{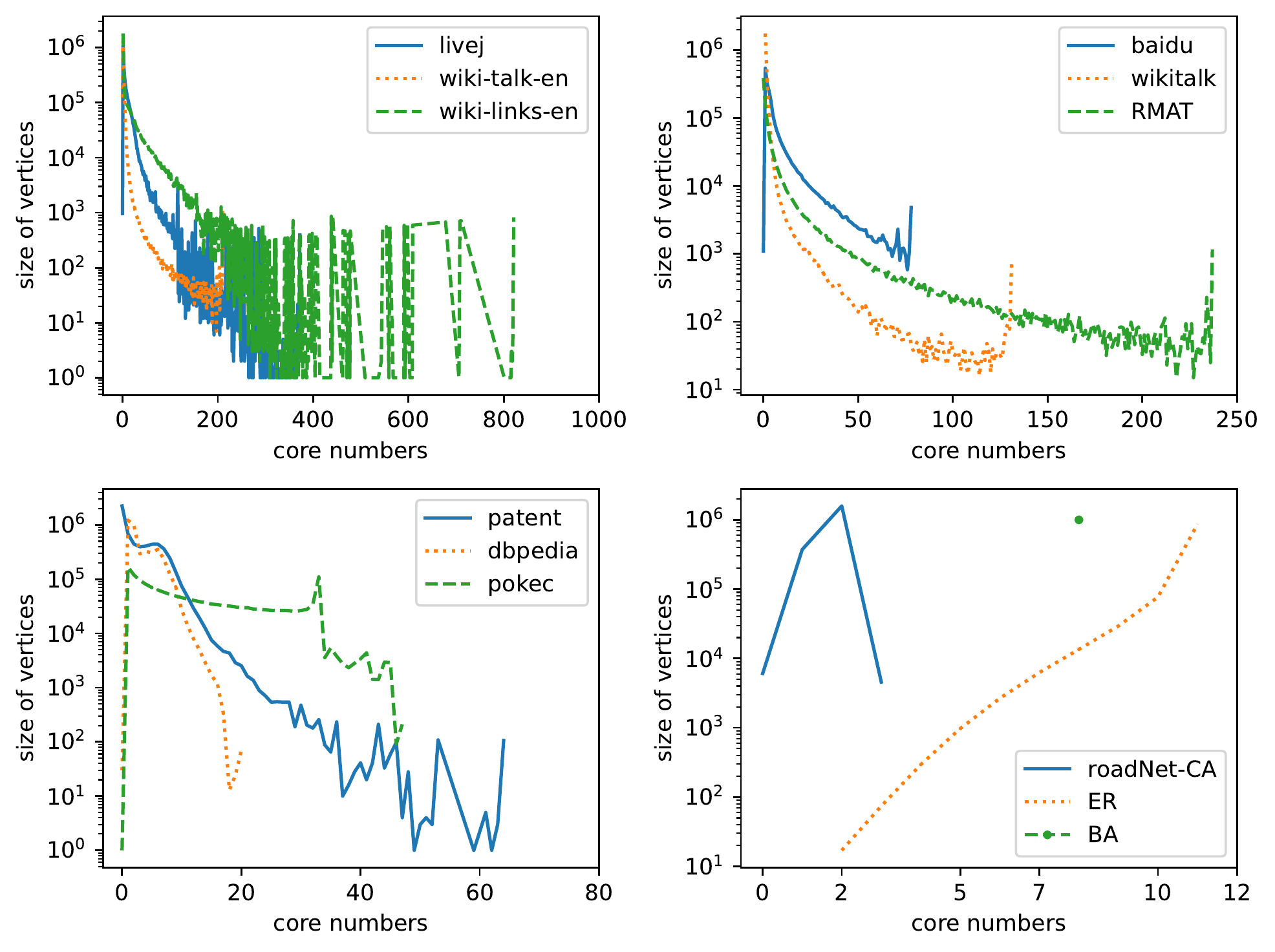}
\centering
\caption{The distribution of core numbers.}
\label{fig:graph}
\end{figure}

\subsection{Running Time Evaluation}
In this experiment, we compare the performance of our simplified order-based method (\texttt{OurI} and \texttt{OurR}) with the original order-based method (\texttt{I} and \texttt{R}). For each tested graph, we first randomly select 100,000 edges out of each tested graph. 
For each graph, we measure the accumulated time for inserting or removal these 100,000 edges. 
Each test runs at least 50 times, and we calculate the means with 95\% confidence intervals.

The results for edge insertion are shown in Figure \ref{fig:time}(a). We can see \texttt{OurI} outperforms \texttt{I} over all tested graphs. Specifically, Table \ref{tb:time} shows the speedups of \texttt{OurI} vs. \texttt{I}, which ranges from 1.29 to 7.69. 
The reason is that the sequence $\od_k$ in $k$-order is maintained separately for each core number $k$.
Each time insert $v$ into or remove $v$ from $\od_k$, \texttt{OurI} only requires worst-case $O(1)$ amortized time while \texttt{I} requires worst-case $O(\log |\od_k|)$ time. 
Therefore, over \emph{BA} we can see \texttt{OurI} gains the largest speedup as $7.69$ since all vertices have single one core number with $|\od_8| = 8,000,000$; over \emph{wiki-links-en} we can see \texttt{OurI} gains the smallest speedup as $1.29$ since vertices have core numbers ranging from $0$ to $821$ such that a large portion of order lists has $|\od_k|$ around $1000$. 

\begin{figure}[!htb]
\centering
\begin{subfigure}[b]{0.45\textwidth}
\centering
\includegraphics[scale=0.4]{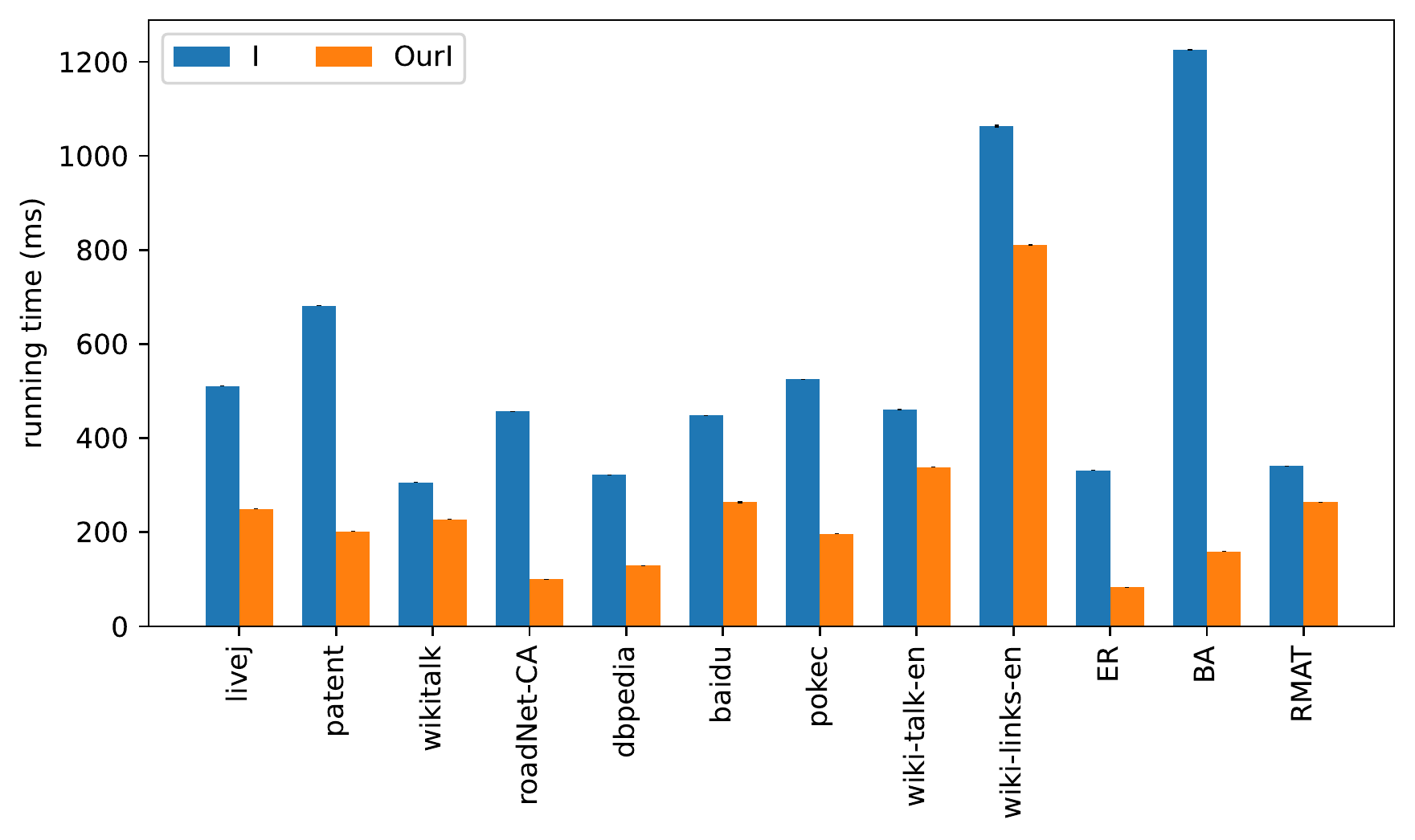}
\caption{Edge Insertion}
\label{fig:time-insert}
\end{subfigure}
\begin{subfigure}[b]{0.45\textwidth}
\centering
\includegraphics[scale=0.4]{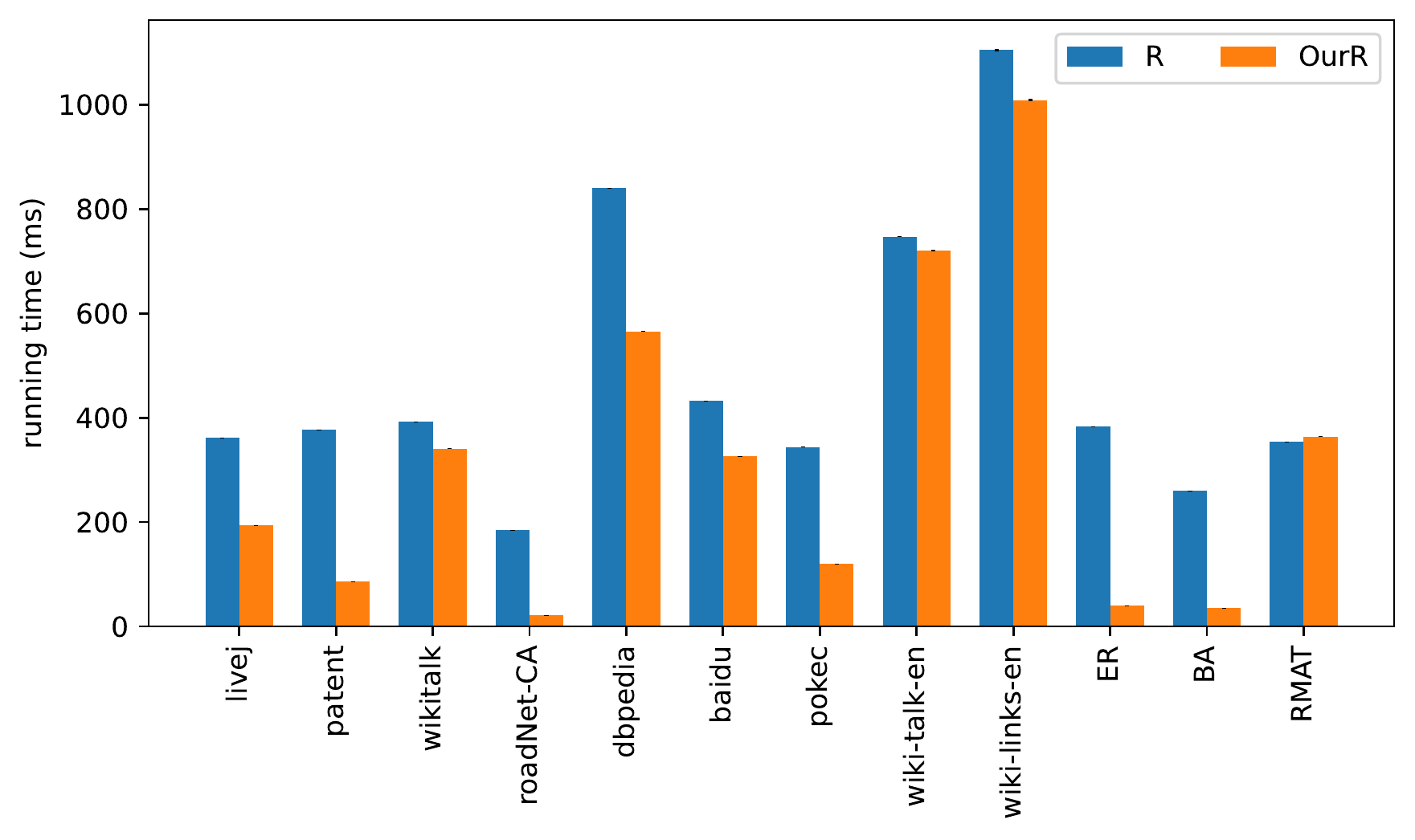}
\caption{Edge Removal}
\label{fig:time-remove}
\end{subfigure}
\caption{Compare the running times of two methods.}
\label{fig:time}
\end{figure}

Similarly, we can see that the edge removal has almost the same trend of speedups in Figure \ref{fig:time}(b), which ranges from $1.16$ to $5.26$ in Table \ref{tb:time}. However, we observe that the speedups of removal may be less than the insertion over most of the graphs. 
The reason is that the edge removal requires fewer order operations compared with the edge insertion. That is, unlike the edge insertion, it is not necessary to compare the $k$-order for two vertices by \texttt{ORDER}$(\od,x,y)$ when reversing the vertices.
The main order operations are \texttt{REMOVE}$(\od, x)$ and then \texttt{INSERT}$(\od,x, y)$, when the core numbers of vertices $x\in V^*$ are off by $1$.

In Table \ref{tb:time}, for the batch insertion, we can see the speedups of \texttt{OurBI} vs. \texttt{I} is much less then the speedups of \texttt{OurI} vs. \texttt{I}, although \texttt{OurBI} may have smaller size of $V^+$ than \texttt{OurI}. 
One reason is that \texttt{OurBI} have to traverse inserted graph $\Delta G$ at most $\DEG(\Delta G)$ round, which is the maximum degree of the inserted graph $\Delta G$. 
The other reason is that compared with \texttt{OurI}, \texttt{OurBI} has larger size of priority queue $Q$, which \texttt{OurBI} requires more running time on enqueue and dequeue operations. 

In Table \ref{tb:time}, we also observe that the speedups of \texttt{OurInit} vs. \texttt{Init} is a little larger than $1$. The reason is that for initialization, most of the running time is spent on computing the core number for all vertices by the BZ algorithm. After running the BZ algorithm, \texttt{OurInit} assigns labels for all vertices to construct $\od$ in $k$-order, which requires worst-case $O(n)$ time. However, \texttt{Init} has to add all vertices to binary search trees, which requires worst-case $O(n\log n)$ time.

\begin{table}[!htb]
\caption{Compare the speedups of our method for all graphs.}
\label{tb:time}
\small
\begin{tabular}{l|cccc}
\toprule
Graph & \texttt{OurI} vs \texttt{I} & \texttt{OurBI} vs \texttt{I} & \texttt{OurR} vs \texttt{R}  & \texttt{OurInit} vs \texttt{Init}\\ 
\midrule
livej & 2.04 & 1.66 & 1.87 & 1.02 \\
patent & 3.37 & 2.68 & 4.41 & 1.04 \\
wikitalk & 1.34 & 1.63 & 1.15 & 1.26 \\
roadNet-CA & 4.51 & 2.95 & 8.56 & 1.17 \\ \hline
dbpedia & 2.49 & 2.14 & 1.49 & 1.08 \\
baidu & 1.70 & 1.68 & 1.33 & 1.04 \\
pokec & 2.67 & 2.37 & 2.87 & 1.03 \\
wiki-talk-en & 1.36 & 1.45 & 1.04 & 1.20 \\
wiki-links-en & 1.31 & 1.16 & 1.09 & 1.02 \\ \hline
ER & 3.97 & 2.76 & 9.72 & 1.08 \\
BA & 7.69 & 5.26 & 7.42 & 1.15 \\
RMAT & 1.29 & 1.31 & 0.97 & 1.09 \\
\bottomrule
\end{tabular}
\end{table}

\subsection{Stability Evaluation}
We test the stability of different methods over two selected graphs, i.e. \emph{wikitalk} and \emph{dbpedia} as follows. 
First, we randomly sample $5,000,000$ edges and partition into $50$ groups, where each group has totally different $100,000$ edges. Second, for each group, we measure the accumulated running time of different methods. That is, the experiments run $50$ times and each time has totally different inserted or removed edges. 

Figure \ref{fig:stability} shows the results over two selected graphs. We can see that \texttt{OurI} and \texttt{OurR} outperform \texttt{I} and \texttt{R}, respectively. More important, the performance of \texttt{OurI} and \texttt{OurR} is as well-bounded as \texttt{I} and \texttt{R}, respectively.
The reason is that \texttt{I} and \texttt{R} are well-bounded as the variation of $V^+$ is small for different inserting or removal edges; also, \texttt{OurI} and \texttt{OurR} have the same size of traversed vertices $V^+$ and thus have similar well-bounded performance. 
\begin{figure}[htb]
\centering
\includegraphics[scale=0.4]{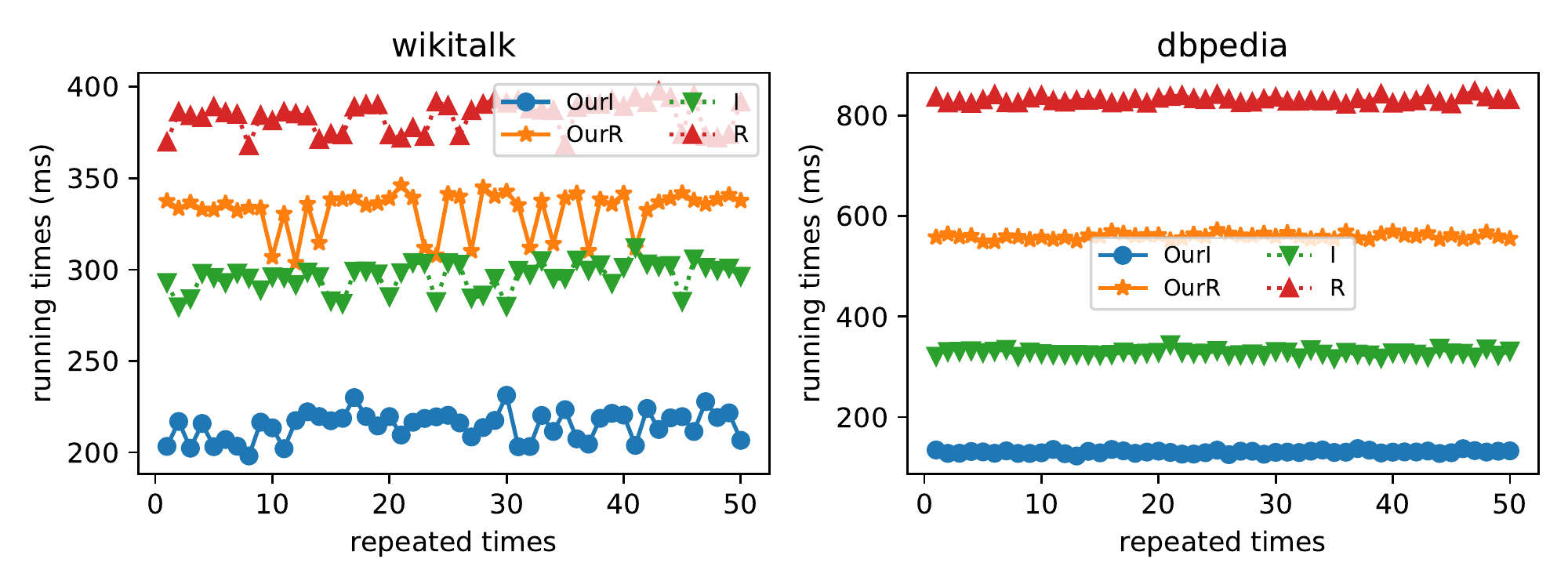}
\centering
\caption{The stability of all methods over selected graphs.}
\label{fig:stability}
\end{figure}

\subsection{Scalability Evaluation}

We test the scalability of different methods over two selected graphs, i.e., \emph{wikitalk} and \emph{dbpedia}. 
We vary the number of edges exponentially by randomly sampling from 100,000 to 200,000, 400,000, 800,000, 1,600,000, etc. We keep incident vertices of edges for each sampling to generate the induced subgraphs. 
Over each subgraph, we further randomly selected 100,000 edges for insertion or removal. 
For example, over \emph{wikitalk}, the first subgraph has 100,000 edges, all of which can be inserted or removed; the last subgraph has 3,200,000 edges, only 100,000 of which can be inserted or removed.
Over each subgraph, we measure the accumulated time for insertion or removal these 100,000 edges. Each test runs at least 50 times, and we calculate the average running time. 

We show the result in Figure \ref{fig:scala}, where the x-axis is the number of sampled edges in subgraphs increasing exponentially, and the y-axis is the running times (ms) for different methods by inserting or removing 100,000 edges. 
Table \ref{tb:scalability} shows the details of scalability evaluation, where $m'$ is the number of edges in subgraphs, $\nlb$~is the number of updated labels used by the Order Data Structures of our methods, and $\nrp$~is the number of outer while-loop repeated rounds for \texttt{OurBI}. 
We make several observations as follows: 

- In Figure \ref{fig:scala}, a first look reveals that the running time of \texttt{OurI} grow more slowly compared with \texttt{I}. The reason is that \texttt{OurI} improves the worst-case running time of each order operation of $\od_k$ from $O(\log |\od_k|)$ to $O(1)$. In this case, the larger sampled graphs have the larger size of $\od_k$, which can lead to the higher speedups.
However, we can see that the running time of \texttt{OurR} always grows with a similar trend compared with \texttt{R}. The reason is that a large percentage of the running time is spent on removing edges from the adjacent lists of vertices, which requires traversing all the corresponding edges. Because of this, even \texttt{OurR} has more efficient order operations for $\od$ than \texttt{R}, the speedups are not obvious.

- In Figure \ref{fig:scala}, we observe that \texttt{OurBI} sometimes runs faster then \texttt{OurI}. The reason is as follow. From Table \ref{tb:scalability}, compared with \texttt{OurI}, \texttt{OurBI} has less traversed vertices ($V^+$), as some repeated traversed vertices can be avoided; also, \texttt{OurBI} has less number of updated labels ($\nlb$), as the number of relabel process can be reduced. 
However, compared with \texttt{OurI}, \texttt{OurBI} may add much more vertices into priority queue $Q$, which cost more running time of enqueue and dequeue. Also, \texttt{OurBI} may require several times of repeated rounds ($\nrp$), which may cost extra running time. 
Typically, this extra running time is acceptable as most of the edges can be inserted in the first round, e.g., for \emph{dbpedia} with $6.4$M sampled edges, the number of batch inserted edges is 100000, 1555, 12, and 0  in four rounds, respectively. 
This is why \texttt{OurBI} sometimes runs faster but sometimes slower compared with \texttt{OurI}.

- In Figure \ref{fig:scala}, we observe that \texttt{OurR} is always faster than \texttt{OurI}. The reason is as follow. From Table \ref{tb:scalability}, compared with \texttt{OurI}, \texttt{OurR} has less number of traversed vertices ($V^*$), as \texttt{OurR} has $V^*=V^+$; \texttt{OurR} has less number of updated labels ($\nlb$), as vertices are removed from $O_{K}$ and then appended after $O_{K-1}$ and thus the relabel process is not always triggered.

\begin{figure}[htb]
\centering
\includegraphics[scale=0.4]{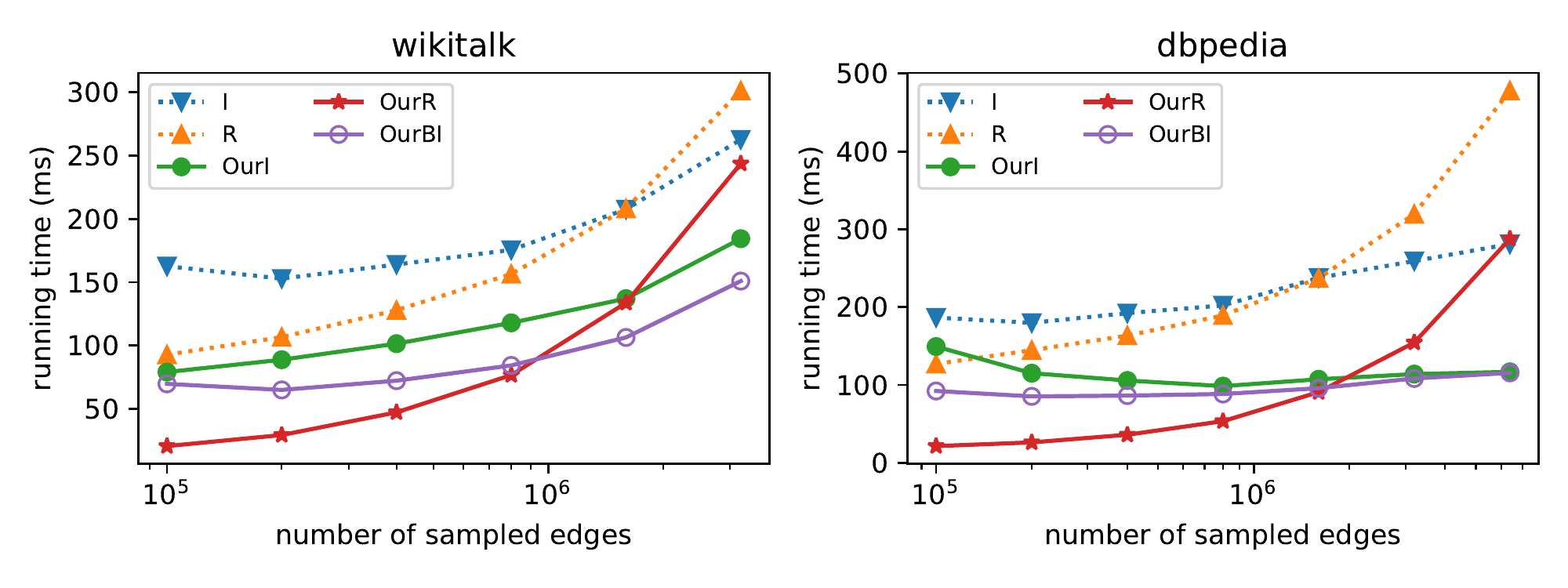}
\centering
\caption{The scalability of all methods over selected graphs.}
\label{fig:scala}
\end{figure}

\iftrue
\begin{table}[!htb]
\caption{The details of scalability evaluation by varying the number of sampled edges over \emph{wikitalk} and \emph{dbpedia}.}
\label{tb:scalability}
\centering
\footnotesize
\begin{tabular}{r|rrr|rrrr|rr}
\toprule
       &\multicolumn{3}{c}{\texttt{OurI}} & \multicolumn{4}{|c}{\texttt{OurBI}} & \multicolumn{2}{|c}{\texttt{OurR}}  \\
 $m'$ & $|V^*|$ & $|V^+|$ & \nlb & $|V^*|$ & $|V^+|$ & \nlb & \nrp & $|V^*|$ & \nlb \\ 
\midrule
0.1M & 107K & 131K & 1.5M & 107K & 123K  & 116K & 10 & 107K & 107K \\
0.2M & 101K & 118K & 1.4M & 101K & 109K  & 111K & 11 & 101K & 101K \\
0.4M & 101K & 116K & 1.3M & 101K & 107K  & 107K & 9  & 101K & 101K \\
0.8M & 101K & 113K & 1.3M & 101K & 106K  & 104K & 10 & 101K & 101K \\
1.6M & 100K & 110K & 1.2M & 100K & 106K  & 103K & 11 & 100K & 100K \\
3.2M & 101K & 110K & 1.1M & 101K & 106K  & 103K & 9  & 101K & 101K \\
\midrule

0.1M & 146K & 149K & 2.6M & 146K & 151K & 149K & 4 & 146K & 146K \\
0.2M & 129K & 136K & 2.1M & 129K & 136K & 132K & 3 & 129K & 129K \\
0.4M & 117K & 131K & 1.8M & 117K & 130K & 124K & 3 & 117K & 117K \\
0.8M & 109K & 127K & 1.6M & 109K & 126K & 118K & 3 & 109K & 109K \\
1.6M & 105K & 127K & 1.4M & 105K & 125K & 116K & 4 & 105K &  105K \\
3.2M & 102K & 124K & 1.3M & 102K & 122K & 113K & 4 & 102K &  102K \\
6.4M & 100K & 124K & 1.2M & 100K & 122K & 112K & 4 & 100K &  100K \\
\bottomrule
\end{tabular}
\end{table}
\fi

\section{Conclusion And Future Work}
In this work, we study maintaining the $k$-core of graphs when inserting or removing edges. We simplify the state-of-the-art core maintenance algorithm and also improve its worst-case time complexity by introducing the classical Order Data Structure. Our simplified approach is easy to understand, implement, and argue the correctness. 
The experiments show that our approach significantly outperforms the existing methods. 

In future work, we can parallelize our approach to run on multi-core machines. The key issue is how to implement the parallel Order Data Structure efficiently. In particular, we can apply the core maintenance to large data graphs with billions of vertices like social networks. 


\bibliographystyle{ACM-Reference-Format}
\bibliography{references}
\end{document}